\documentclass[11pt]{article}
\usepackage{fullpage,amssymb,amsmath}
\usepackage{graphicx,algorithmicx,algpseudocode,algorithm}
\usepackage{enumerate}
\usepackage{tikz}
\usepackage[mathscr]{euscript}
\usepackage{calrsfs}
\usetikzlibrary{shapes}

\usepackage{enumitem}
\usepackage{geometry}
\usepackage{prerex}
\usepackage{verbatim}

\usepackage{xspace}

 \usepackage{xcolor}
 \usepackage{mathtools}
\usepackage{microtype}
\usepackage{amsfonts}
\usepackage{comment}
\usepackage[english]{babel}

\usepackage{enumitem}
\usepackage{todonotes}
\renewcommand{\todo}[1]{}
\usepackage{longtable}

\definecolor{blue}{rgb}{0.1,0.2,0.5}
\definecolor{brown}{rgb}{0.6,0.6,0.2}
\usepackage{hyperref}
\usepackage{cleveref}
\usepackage[amsmath,thmmarks,hyperref]{ntheorem}
\usepackage{scalefnt}

\crefformat{page}{#2page~#1#3}%
\Crefformat{page}{#2Page~#1#3}%
\crefformat{equation}{#2(#1)#3}%
\Crefformat{equation}{#2(#1)#3}%
\crefformat{figure}{#2Figure~#1#3}%
\Crefformat{figure}{#2Figure~#1#3}%
\crefformat{section}{#2Section~#1#3}
\Crefformat{section}{#2Section~#1#3}
\crefformat{chapter}{#2Chapter~#1#3}
\Crefformat{chapter}{#2Chapter~#1#3}
\crefformat{chapter*}{#2Chapter~#1#3}
\Crefformat{chapter*}{#2Chapter~#1#3}
\crefformat{part}{#2Part~#1#3}
\Crefformat{part}{#2Part~#1#3}
\crefformat{enumi}{#2(#1)#3}
\Crefformat{enumi}{#2(#1)#3}

\usepackage{enumerate}

\usepackage{latexsym}

\theoremnumbering{arabic}
\theoremstyle{plain}
\theoremsymbol{}
\theorembodyfont{\itshape}
\theoremheaderfont{\normalfont\bfseries}
\theoremseparator{.}

\newtheorem{theorem}{Theorem}
\crefformat{theorem}{#2Theorem~#1#3}
\Crefformat{theorem}{#2Theorem~#1#3}

\newcommand{\newtheoremwithcrefformat}[2]{%
  \newtheorem{#1}[theorem]{#2}%
  \crefformat{#1}{##2\MakeUppercase#1~##1##3}%
  \Crefformat{#1}{##2\MakeUppercase#1~##1##3}%
}
\newcommand{\newseptheoremwithcrefformat}[2]{%
  \newtheorem{#1}{#2}%
  \crefformat{#1}{##2\MakeUppercase#1~##1##3}%
  \Crefformat{#1}{##2\MakeUppercase#1~##1##3}%
}

\newtheoremwithcrefformat{lemma}{Lemma}
\newseptheoremwithcrefformat{amend}{Amendment}
\newtheoremwithcrefformat{fact}{Fact}
\newtheoremwithcrefformat{proposition}{Proposition}
\newtheoremwithcrefformat{observation}{Observation}
\newtheoremwithcrefformat{conjecture}{Conjecture}
\newtheoremwithcrefformat{corollary}{Corollary}
\newseptheoremwithcrefformat{claim}{Claim}
\theorembodyfont{\upshape}
\newtheoremwithcrefformat{example}{Example}
\newtheoremwithcrefformat{remark}{Remark}
\newseptheoremwithcrefformat{definition}{Definition}

\theoremstyle{nonumberplain}
\theoremheaderfont{\scshape}
\theorembodyfont{\normalfont}
\theoremsymbol{\ensuremath{\square}}
\newtheorem{proof}{Proof}

\theoremsymbol{\ensuremath{\lrcorner}}

\newcommand\ScaleExists[1]{\vcenter{\hbox{\scalefont{#1}$\exists$}}}

\DeclareMathOperator*\bigexists{%
  \vphantom\sum
  \mathchoice{\ScaleExists{2}}{\ScaleExists{1.4}}{\ScaleExists{1}}{\ScaleExists{0.75}}}
\renewcommand{\ell}l

\def\cqedsymbol{\ifmmode$\lrcorner$\else{\unskip\nobreak\hfil
\penalty50\hskip1em\null\nobreak\hfil$\lrcorner$
\parfillskip=0pt\finalhyphendemerits=0\endgraf}\fi}

\newcommand{\Q}{\mathbb Q}

\newcommand{\set}[1]{\{#1\}}
\newcommand{\setof}[2]{\set{#1\mid#2}}

\newcommand{\wh}{\widehat}
\newcommand{\from}{\colon}
\newcommand{\tup}[1]{\mathbf{#1}}
\renewcommand{\subset}{\subseteq}

\renewcommand{\phi}{\varphi}
\newcommand{\str}[1]{\mathscr{#1}}
\newcommand{\whstr}[1]{{\wh{\str #1}}}

\newcommand{\Oof}{\mathscr{O}}

\newcommand{\CCC}{\mathcal{C}}

\newcommand{\DDD}{\mathcal{D}}

\newcommand{\FFF}{\mathcal{F}}
\newcommand{\GGG}{\mathcal{G}}

\newcommand{\ind}[2][]{%
  \mathrel{
    \mathop{
      \vcenter{
        \hbox{\oalign{\noalign{\kern-.3ex}\hfil$\vert$\hfil\cr
              \noalign{\kern-.7ex}
              $\smile$\cr\noalign{\kern-.3ex}}}
      }
    }^{#2}\displaylimits_{#1}
  }
}

\newcommand{\eps}{\varepsilon}

\newcommand{\ttup}[1]{}
\newcommand{\DKT}{Dvo\v r\'ak, Kr\'al', and Thomas\xspace}

\newcommand{\N}{\mathbb{N}}
\newcommand{\Z}{\mathbb{Z}}
\newcommand{\R}{\mathbb{R}}
\renewcommand{\phi}{\varphi}
\renewcommand{\epsilon}{\varepsilon}

\renewcommand{\mid}{~:~}
\newcommand{\sem}[1]{\mathbf{#1}}

\DeclareMathOperator{\perm}{perm}

\renewcommand{\ge}{\geqslant}

\renewcommand{\le}{\leqslant}

\newcommand{\flt}{\oplus}

\newcommand{\parent}{\mathsf{parent}}
\newcommand{\troot}{\mathsf{root}}
\newcommand{\rt}{\mathsf{root}}
\newcommand{\ops}{\mathbb{C}}

\newcommand{\fops}{\mathrm{FO[\ops]}}
\newcommand{\foops}{\mathrm{FO_G}[\ops]}
\newcommand{\GS}{Grohe and Schweikardt\xspace}
\newcommand{\appmark}{$\ast$}

\begin{document}
\title{Aggregate Queries on Sparse Databases}

\author{
Szymon Toru\'nczyk \thanks{
  Institute of Informatics, University of Warsaw, Poland, \texttt{szymtor@mimuw.edu.pl}.}
}

\newcommand{\ifdoublecolumn}[1]{}
\newcommand{\ifsinglecolumn}[1]{#1}

\maketitle

\begin{abstract}
We propose an algebraic framework for studying efficient 
algorithms for query evaluation,  aggregation, 
enumeration, and maintenance under updates, on sparse  
databases. Our framework allows to treat those problems 
in a unified way, by considering various semirings, 
depending on the considered problem. As a concrete 
application, we propose a powerful query language 
extending first-order logic by aggregation in multiple 
semirings. We obtain an optimal algorithm for computing 
the answers of such queries on sparse databases. More 
precisely, given a database from a fixed class with \emph
{bounded expansion}, the algorithm computes in linear time
a data structure which allows to enumerate the set of 
answers to the query, with constant delay between two 
outputs.


\end{abstract}


\section{Introduction}

The central focus of database theory is to understand the
complexity of evaluating queries under various assumptions  
on the query and the data.
There are two long lines of research in this direction. 
In the first line of work, the aim is to restrict the structure of the query, as measured using various width measures, such as fractional hyper-tree width or submodular width. 
In another line of work, pursued in this paper, structural restrictions are imposed on the data.
Employing ideas from graph theory, this leads to algorithms with linear data complexity, assuming the database is sparse in some sense. Under sufficiently 
strong assumptions on the sparsity of the data, captured by the notion of \emph{bounded expansion}, such linear-time algorithms can be obtained~\cite{DBLP:conf/focs/DvorakKT10,Dvorak2013testing} for all queries in first-order logic (equivalently, relational algebra).

In this paper, we consider two query languages, called \emph{weighted queries} and \emph{nested weighted queries}, which extend first-order logic by the ability of performing aggregation using counting, summation, minimum, and in fact, summation in arbitrary \emph{commutative semirings}. 
An example of a {weighted query} is:
\begin{align*}
f=\sum_{x,y,z}[E(x,y)\land E(y,z)\land E(z,x)]\cdot w(x,y)\cdot w(y,z)\cdot w(z,x).
\end{align*} Here $E$ is a binary edge relation and $w$ is a binary weight symbol.
This query can be evaluated in a (directed) graph $G=(V,E)$ equipped with a binary weight function $w\from V\times V\to \sem S$ taking values in an arbitrary semiring $\sem S$ (all semirings are assumed to be commutative in this paper).
The operator $[\cdot]$ maps the boolean \emph{true} to $1\in \sem S$ and \emph{false} to $0\in\sem S$.
If $\sem S$ is the semiring $(\N,+,\cdot)$ and $w(a,b)$ represents the \emph{multiplicity} of an edge $(a,b)\in E$,
then the value of the query corresponds to the \emph{bag semantics} of the query $\phi(x,y,z)=E(x,y)\land E(y,z)\land E(z,x)$, i.e., the number of directed triangles in $G$, treated as a multigraph.
 If $\sem S$ is the semiring $(\N\cup\set{+\infty},\min,+)$ with $\min$ playing the role of addition and $+$ of multiplication, and $w(a,b)$ represents the \emph{cost} of an edge $(a,b)\in E$,
 then $f$ evaluates to the minimum total cost $w(a,b)+w(b,c)+w(c,a)$ of a directed triangle in $G$.

In general, a \emph{weighted query} $f$ can use arbitrary \emph{first-order formulas} inside the brackets $[\cdot]$, and
  the semiring operations $+$ and $\cdot$ and aggregation (summation) $\sum$, applied to weight symbols or expressions $[\cdot]$. This corresponds exactly to the \emph{positive relational algebra} of Green, Karvounarakis and Tannen~\cite{provenance}, where we allow arbitrary first-order selection predicates.

We also consider more general \emph{nested weighted queries} which allow using multiple semirings within one query.  An example of such a query is $$\max_{x} \left(\sum_y{[E(x,y)]_{\N}\cdot w(y)}\right)/ \left(\sum_y{[E(x,y)]_{\N}}\right),$$ which involves the semiring  $(\N,+,\cdot)$, used within the parenthesis, and $\sem Q_{\textrm{max}}=(\Q\cup\set{-\infty},\max,+)$, with $\max$ playing the role of addition and $+$ of multiplication, used in the outermost aggregation. Finally, $/\from \N\times \N\to \Q$ is a connective denoting division. Given a graph $G=(V,E)$ with a unary weight function $w\from V\to \N$,
the query computes the maximum over all vertices $x$ of the average weight of the neighbors of $x$.
Another example\todo{improve} of a nested weighted query is the following query with boolean output:
$$f(x)=\exists_{y} 
E(x,y)\land 
\big(w(y)>\sum_z [E(y,z)]_\N\cdot w(z)\big),$$
which determines if a given node $x$ has some neighbor $y$ whose weight is larger than sum of the weights its neighbors.
Here, existential quantification $\exists$ is just summation in the boolean semiring $\sem B=(\set{\textit{true, false}},\lor,\land)$, and 
the inequality $>$ is treated as a connective $> \from \N\times \N\to \sem B$ comparing values.
Nested weighted queries extend the logic $\textrm{FOC}(\mathbb P)$ introduced by Kuske and Schweikardt~\cite{Kuske:2017:FLC:3329995.3330068}. 

We now describe our main results concerning the evaluation of such queries on sparse databases. Our main conceptual contribution is a framework based on circuits over semirings. A key result in this framework is an algorithm computing certain circuits
representing the output of a query. As concrete applications, we obtain an efficient dynamic algorithm for evaluating weighted queries and 
an efficient static algorithm for evaluating weighted nested queries on sparse databases.
To better explain and motivate our framework and its key result, we start with  describing its applications.

\newcounter{para}
\renewcommand{\thepara}{\Alph{para}}

\newcommand\myparagraph[1]{\refstepcounter{para}\paragraph{{\normalfont(\thepara)\ }#1}}

\myparagraph{Evaluating weighted queries on sparse databases.}\label{res:A}


   Fix a weighted query $f$ and a database $\str D$ equipped with several weight functions taking values in a semiring $\sem S$.
  The goal is to compute the value of the query $f$ on $\str D$. We focus on the data complexity and assume the query $f$ to be fixed.

Our result gives an algorithm with \emph{linear-time} data complexity for \emph{sparse} input databases $\str D$.
If $f(\tup x)$ is a weighted query with free variables, then the linear-time algorithm computes a data structure allowing to query the value $f(\tup a)$, given a tuple $\tup a\in \str D^{\tup x}$ in time \emph{logarithmic } in $\str D$ in general, and in \emph{constant time}, if the semiring $\sem S$ is a ring or is finite.
Furthermore, our algorithm is dynamic, and allows maintaining the data structure in logarithmic time  whenever a weight is updated. 

   Our notion of sparseness is based on the \emph{Gaifman graph} of the database $\str D$, representing which elements of the database occur together in a tuple in some relation.
   We assume that the Gaifman graph of the input database $\str D$ belongs to a fixed class of graphs which has \emph{bounded expansion}. Classes of bounded expansion include classes of bounded maximum degree, the class of planar graphs, or every class which excludes a fixed topological minor.


  \myparagraph{Evaluating nested weighted queries on sparse databases.}\label{res:B}
For a broad class of {nested} weighted queries, denoted $\foops$, we also achieve linear-time data complexity for databasese from a fixed class with bounded expansion.
   The query language $\foops$ includes the examples of nested weighted queries mentioned above,
   and
   vastly extends a query language introduced by Grohe and Schweikardt, for which they also obtain linear-time evaluation  on classes with bounded expansion, and  almost linear-time evaluation for the slightly more general \emph{nowhere dense} classes~\cite{DBLP:conf/pods/GroheS18}.

\myparagraph{Provenance analysis.}\label{res:C}
In the two results above, we assumed the \emph{unit cost model},
where semiring elements can be stored in a single memory cell, and semiring operations
take constant time. Our semiring framework also 
allows considering complex representations of semiring elements. This has applications in provenance analysis and query enumeration, as illustrated below.

One of the aims in \emph{provenance analysis} is to analyse how the tuples in the input database  contribute to the answers to a query.
Consider for example the query $\psi(x)=\exists_{y,z}E(x,y)\land E(y,z)\land E(z,x)$.
Given a graph $G=(V,E)$, assign to each edge $(a,b)\in E$ a unique identifier $e_{ab}$.
Let $\sem S$ be the \emph{free commutative semiring} (also called the \emph{provenance semiring}), consisting of all formal sums
of products of such identifiers, e.g. $e_{ab}\cdot e_{bc}\cdot e_{ca}+e_{ac}\cdot e_{cd}\cdot e_{da}$.
Define the binary  weight function ${w\from V\times V\to \sem S}$ by setting $w(a,b)=e_{ab}$.
Then the weighted expression $f(x)=\sum_{y,z}w(x,y)\cdot w(y,z)\cdot w(z,x)$ evaluates at a node $a$ to the \emph{provenance of $a$}, i.e., the element of $\sem S$ which is the sum of all products $e_{ab}\cdot e_{bc}\cdot e_{ca}$, where $abc$ is a directed triangle in $G$.


As each element of $\sem S$ is a formal sum of possibly unbounded length, we therefore represent each such element by an enumerator which enumerates the elements of the formal sum. For a weighted query $f(\tup x)$ and input database $\str D$ whose weights  are represented by as such black box enumerators,
we show how to compute in linear data complexity a data structure which, given a tuple~$\tup a\in \str D^{\tup x}$, outputs an enumerator for $f_\str D(\tup a)$. The maximal delay between two outputs of the resulting enumerator is asymptotically equal to the maximal delay of the input enumerators.
This yields applications to query enumeration, as follows.

\myparagraph{Dynamic query enumeration.}\label{res:D}
Instead of analysing the tuples in the input database which contribute to the query answers,
the free semiring may be used in a slightly different way, to represent the set of answers to a first-order query $\phi(\tup x)$ in a given database. For example, if $\tup x=\set{x,y,z}$,
consider the weighted expression $f'=\sum_{x,y,z} [\phi(x,y,z)]\cdot w_1(x)\cdot w_2(y)\cdot w_3(z)$, where for $i=1,2,3$ the weight function $w_i$  assigns  a unique identifier $e^i_a$ to each element $a$ in the active domain.
Then $f'$ evaluates to a formal sum of products $e_a^1\cdot e_b^2\cdot e^3_c$, over all answers $(a,b,c)$ to $\phi$. Applying the result~\eqref{res:C} yields a linear-time algorithm computing a constant-delay enumerator for the set  of answers of a fixed first-order formula  in a given input database.
This reproves a result of Kazana and Segoufin~\cite{KazanaS13},
and strengthens it by allowing updates to $\str D$ which preserve its Gaifman graph.

\ifdoublecolumn{
\begin{figure*}[h!]
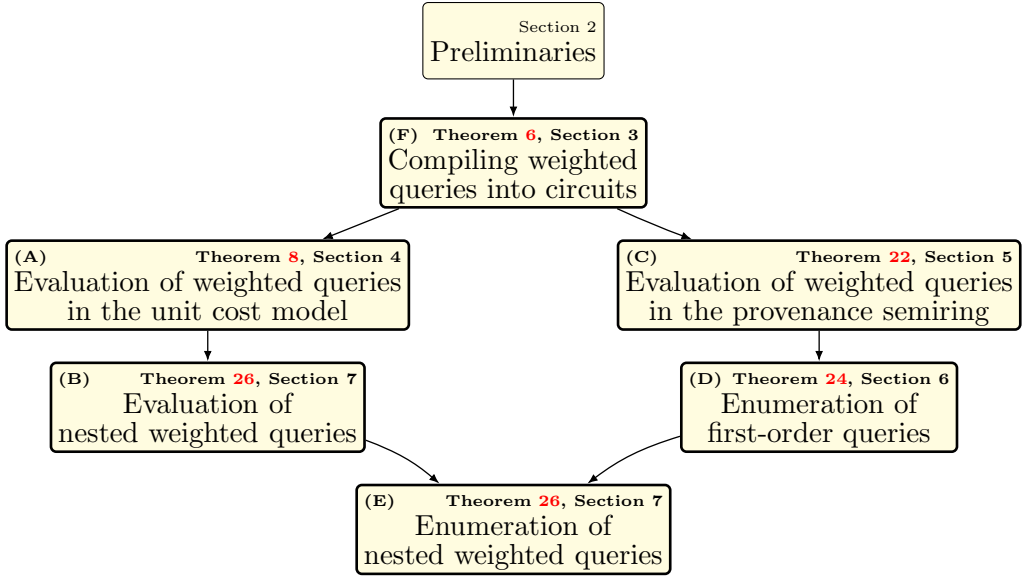

  \centering	
  \include{diagram}
  \caption{The results in the paper and  dependencies between them. The letter in the upper-left corner corresponds to the enumeration in the introduction, and the corresponding theorem is indicated in the upper-right corner.
  } 
	\label{fig:org}
\end{figure*}
}

\myparagraph{Enumerating answers to nested weighted queries}\label{res:BD}
Our framework allows to easily combine the enumeration algorithm~\eqref{res:D} with the result~\eqref{res:B},
yielding a linear time algorithm which computes a constant-delay enumerator for the set of answers to a boolean-valued nested weighted query $\phi(\tup x)$. As a single result, this is the crowning achievement of this paper (it does not subsume the other results in the paper, however). It gives an optimal -- linear in the size of the input and output -- algorithm for answering queries in an expressive language allowing  aggregates in multiple semirings, on sparse databases.


\myparagraph{Circuits with permanent gates.}\label{res:E}
Perhaps the most important contribution of this paper is
of a more conceptual nature. It provides a unified framework, based on \emph{circuits over semirings}.
This framework allows to capture the complexity of evaluating weighted queries on sparse databases, in a way which is independent of the chosen semiring and the representation of its elements. By plugging in specific semirings,
we easily obtain the above results, and more.



Besides usual addition and multiplication gates, our circuits
have gates corresponding to the \emph{permanent} of a rectangular matrix of semiring elements.
For example, 
$$\perm \left[\begin{array}{ccccc}
  a_1&a_2&\ldots&a_{n-1}&a_n\\
  b_1&b_2&\ldots&b_{n-1}&b_n\\
  c_1&c_2&\ldots&c_{n-1}&c_n
\end{array}\right]\quad=\sum_{\substack{i, j, k\\\textit{distinct}}}a_i\cdot  b_j\cdot  c_k,$$
and this is naturally generalized to more  rows.

\medskip
Fix a weighted query $f$ without free variables.
Our key result is an algorithm which compiles $f$ on a given sparse database $\str D$ into a circuit $C_\str D$.
The inputs to the circuit are the weights of the tuples in the database. 
Once the weights are given, the circuit $C_\str D$ evaluates to the value of the expression $f$ on $\str D$ with the given weights.
The circuit is computed in {linear time} from the input database $\str D$, and has {constant}   depth, fan-out, and number of rows in each permanent gate.

\medskip

  
  This result allows to reduce the problem of evaluating arbitrary 
  weighted queries on an arbitrary class of sparse databases to the 
  problem of computing the permanent in the considered semiring $\sem S$ of a given matrix with a {fixed number of rows}, which is a purely algebraic question.
  For example, we show that in any semiring $\sem S$, in the unit cost model, the permanent of a $k\times n$ matrix
  can be computed in time \emph{linear in $n$}, and can be updated in constant time when $\sem S$ is a ring or is finite, and in logarithmic time in general. 
  This yields the result  \eqref{res:A}, where the unit-cost model is assumed. 
  To obtain the result  \eqref{res:C} concerning the free semiring, we show that
   given a $k\times n$ matrix $M$ of elements of the free semiring, each represented  by a constant-delay enumerator,  we can compute in time linear in $n$ a constant-delay enumerator for the permanent of $M$.

\bigskip

\ifsinglecolumn{
\begin{figure*}[h!]
  \centering	
  \include{diagram}
  \caption{The results in the paper and  dependencies between them. The letter in the upper-left corner corresponds to the enumeration in the introduction, and the corresponding theorem is indicated in the upper-right corner.
  } 
	\label{fig:org}
\end{figure*}
}

\paragraph{Summary and organization.}
The organization of the paper is illustrated in Figure~\ref{fig:org}.
Each of the results~\eqref{res:A}-\eqref{res:E} listed above is novel and, we believe, interesting in its own rights. In our view, one of the main contributions of this paper is on a conceptual level, as it provides a unified framework for studying aggregation, enumeration, and updates  based on evaluation of circuits in various semirings, and on reducing the analysis of arbitrary  queries to
the study of the permanent. This framework relies on the  main technical contribution, the result~\eqref{res:E}.
It produces a versatile and universal data structure, namely circuits, which can be used to different effects, by plugging in suitable semirings. 
It enables to easily derive the results~\eqref{res:A}-\eqref{res:D}, culminating in the result~\eqref{res:BD},
which gives an efficient algorithm for enumerating the answers to queries from an expressive query language with semiring aggregates.

Due to lack of space, we can only explore a few applications of our framework in this paper. However, we believe that it offers an explanation of the tractability of various problems
studied  here, and in other papers~\cite{DBLP:conf/focs/DvorakKT10,Dvorak2013testing,KazanaS13,DBLP:journals/tods/BerkholzKS18, Kuske:2017:FLC:3329995.3330068, DBLP:conf/pods/GroheS18}. The proofs of the results marked (\appmark) are moved to the appendix.


\paragraph{Background and related work.}
It is known from the work of \DKT~\cite{DBLP:conf/focs/DvorakKT10,Dvorak2013testing} that for classes with bounded expansion, model-checking first-order logic has linear data complexity. This result has been extended in multiple ways. For example, Kazana and Segoufin~\cite{KazanaS13} obtained an enumeration algorithm, in which a linear-time preprocessing phase is followed by an enumeration phase which outputs all answers to the query with constant-time delay between two outputs. Both results have been extended to the slightly more general \emph{nowhere dense} classes,
replacing linear time by almost-linear time~\cite{gks,DBLP:conf/pods/SchweikardtSV18}. Furthermore, Grohe and Schweikardt~\cite{DBLP:conf/pods/GroheS18} extended the model-checking result to an extension of first-order logic by counting aggregates, denoted FOC$_1$($\mathbb P$). Our results~\eqref{res:B} and~\eqref{res:BD} consider a much more powerful logic, allowing arbitrary semiring aggregates, but on the slightly less general classes with bounded expansion. In particular, we provide a partial answer to a question posed in~\cite[Section 9]{DBLP:conf/pods/GroheS18}, asking if their result can be extended beyond counting aggregates. We  also provide a partial answer to another question posed there, asking about enumerating answers to FOC$_1$($\mathbb P$) queries. Both  answers are positive for classes with bounded expansion.

These results are \emph{static}, i.e., any modification of the database requires recomputing the structure from scratch. Recently,  Berkholz, Keppeler, and Schweikardt~\cite{DBLP:journals/tods/BerkholzKS18} obtained a dynamic version of the result of Kazana and Segoufin, maintaining the data structure in constant time each time it is modified,
 but under the strong restriction that the structures have bounded degree. This result was then generalized by Kuske and Schweikardt~\cite{Kuske:2017:FLC:3329995.3330068} to
the logic with counting FOC$_1$($\mathbb P$). Our result~\eqref{res:D} is only partially dynamic, as it allows updates which preserve the Gaifman graph, but it applies to all classes of bounded expansion. However, for the special case of classes of bounded degree, the fully dynamic version follows easily, thus recovering the result \cite{DBLP:journals/tods/BerkholzKS18} for first-order logic (without counting extensions).


Computing circuits for representing query outputs on sparse graphs has been employed by Amarilli, Bourhis, Mengel~\cite{DBLP:conf/icalp/AmarilliBJM17, DBLP:conf/icdt/AmarilliBM18}. They consider classes of \emph{bounded treewidth}, which are much more restricted than classes of bounded expansion. On the other hand, they consider monadic second order (MSO) queries, which are much more powerful than first-order logic. However, they only consider the boolean semiring, and do not combine the power of MSO with other semirings.


Our circuits are very similar to, and extend,
\emph{deterministic decomposable negation normal forms} (d-DNNF) used in \emph{knowledge compilation}~\cite{DBLP:conf/ijcai/Darwiche99}.
They generalize them by allowing permanent gates, which, in our circuits, turn out to be both disjunctive and decomposable, in an appropriate sense.
Our circuits can be alternatively viewed as \emph{factorized representations} 
(extended by suitable permanent operators)
of query answers~\cite{Olteanu:2012:FRQ:2274576.2274607}.

\paragraph{Acknowledgements.}
I am very grateful to Luc Segoufin for initiating the discussion on the result~\eqref{res:D}, and 
 am especially indebted to Michał Pilipczuk, who contributed to initial work towards that result.
I would also like to thank Eryk Kopczyński, Filip Murlak, Stephan Mengel,    Dan Olteanu, Sebastian Siebertz, Alexandre Vigny, and Thomas Zeume,  for  discussions on related topics. I am also grateful to the anonymous reviewers for useful comments.

\section{Preliminaries}
We recall the notion of classes of bounded expansion and the quantifier elimination result of \DKT.

\subsubsection*{Graph classes of bounded expansion}

We consider undirected and simple graphs.
To unclutter notation, whenever $G$ is a graph, its underlying set of nodes is also denoted $G$. We recall the notions of bounded expansion, treedepth, and low treedepth colorings.





\paragraph{Bounded expansion.}
The concept of classes of bounded expansion, proposed by Ne\v{s}et\v{r}il and Ossona de Mendez in~\cite{nevsetvril2008grad}, captures uniform sparsity for graphs.
In the original definition, a graph class $\GGG$ has bounded expansion if and only if for every $r\in \N$ there exists a constant $c_r$ such that the following holds:
whenever we take a graph $G\in \GGG$ and construct a graph from $G$ by first removing some edges and vertices, and then contracting a collection of disjoint, connected subgraphs of radius at most $r$, then in the resulting graph, the ratio between the number
of edges and the number of vertices is bounded by~$c_r$. Intuitively, this means we require that not only the graph class $\GGG$ is sparse 
in the sense of admitting a linear bound on the number of edges in terms of the number of vertices, but this behavior persists even if one applies local contractions in the graph. 
For example, any class of bounded degree, 
the class of planar graphs, or any class excluding a fixed graph as a minor, has bounded expansion.

It turns out that the property of having bounded expansion is equivalent to a variety of other, seemingly unrelated conditions; see~\cite{sparsity} for an introduction to the topic.
In this work we will use the following  one, which allows to decompose  graphs from a class of bounded expansion  
into subgraphs which are essentially trees of bounded depth. 
This is a powerful decomposition method which allows to reduce various problems concerning classes of bounded expansion to the case of trees of bounded depth.

\paragraph{Low treedepth colorings.}
The \emph{treedepth} of a graph $G$ is the minimal depth of a rooted
forest $F$ with the same vertex set as $G$, such that for every edge
$vw$ of $G$, either $v$ is an ancestor of $w$, or $w$ is an ancestor of $v$
in $F$.\todo{pic}  A class $\GGG$ of graphs has \emph{bounded treedepth} if
there is a bound $d \in \N$ such that every graph in $\GGG$ has
treedepth at most $d$. Equivalently, $\GGG$ has bounded treedepth if
there is some number $k$ such that no graph in $\GGG$ contains a 
path of length~$k$~\cite{sparsity}.

A class $\GGG$ has {\em{low treedepth colorings}} if for every $p\in\N$ there is a number $d\in\N$ and a finite set of colors $C$ 
such that every graph $G\in\GGG$
has a (not necessarily proper) vertex coloring $f\from G\to C$ 
such that for any set $D\subset C$ of at most $p$ colors, the subgraph  $G[D]$ of $G$ induced by $f^{-1}(D)$ has treedepth at most~$d$. 


\begin{proposition}[\cite{nevsetvril2008grad}]\label{prop:covers}
Every class of graphs with bounded expansion has low treedepth colorings.
\end{proposition}
We remark that our notion of low treedepth colorings 
is more relaxed than the one from~\cite{nevsetvril2008grad}, but is sufficient for our needs.
The converse  of Proposition~\ref{prop:covers} also holds~\cite[Lemma 4.4]{lsd},
but will not be needed here.

Let us note that if $\GGG$ has bounded expansion, then for fixed $p\in \N$,  given a graph $G\in \GGG$, a coloring $f\from G\to C$   as in the definition above can be computed in linear time~\cite{nevsetvril2008grad}. 

\begin{example}\label{ex:homo}
We give a basic application of this powerful result, due to~\cite{nevsetvril2008gradb}. The proof of our key result, Theorem~\ref{thm:circuits}, is an extension of this approach.

Suppose we want  to answer a boolean conjunctive query $\phi$ in a given graph $G$ from a class with bounded expansion $\GGG$. We can do this in time $\Oof_{\CCC,\phi}(|G|)$, as follows (here, and later, the subscript in the $\Oof$ notation indicates the parameters on which the hidden constants depend on).
 Let $p$ be the number of variables in $\phi$. Let $f\from G\to C$ be a coloring as above. Then $G\models \phi$ if, and only if, $G[D]\models \phi$ for some $D\subset C$ of size~$p$. Hence, it suffices to answer $\phi$ on each of the $\Oof(|C|^p)=\Oof_{\CCC,\phi}(1)$ subgraphs $G[D]$, which now have treedepth bounded by some constant $d=\Oof_{\CCC,\phi}(1)$. This reduces the problem to the case when $\GGG$ is a class of bounded treedepth. This, in turn, reduces to labeled rooted forests of bounded depth, as follows.

 If  $\GGG$ is a class of graphs of treedepth at most $d$, then every graph in $G$ has maximal path length bounded by some constant $d'$
 (it is easy to show that $d'\le 2^d$).
 Given a graph $G$, compute in linear time a rooted spanning forest $F$ of $G$
 by a depth-first search.
 Then $F$ has depth at most $d'$, and for every edge $vw$ of $G$,
 either $v$ is an ancestor of $w$ in $F$, or vice-versa.
  Extend $F$ by adding unary predicates marking, for each node~$v$, the depth of $v$ in $F$, and the relationship in $G$ between $v$ and each of its ancestors at depth $0,1,\ldots,d'$. More precisely, for each $0\le i\le d'$, add unary predicates $D^{i}$ and $U^i$, where $D^i$  
 marks a node $v$ if it has depth $i$, and $U^i$ marks a node $v$ if the ancestor $w$ of $v$ in $F$ at depth $i$ is a neighbor of $v$ (if $w$ exists). It is not difficult to see that
any existential formula $\phi$ can be rewritten into an equivalent existential formula $\widehat\phi$ over the signature consisting of the predicates $D^0,\ldots,D^{d'},U^{0},\ldots,U^{d'}$ and the parent relation,
 so that $G\models \phi$ if and only if $F\models \widehat\phi$. It therefore remains to solve the model-checking problem for existential formulas on labeled rooted forests of depth bounded by $d'$ in linear time, which can be solved by a dynamic approach. \qed
\end{example}

\subsubsection*{Structures and logic}
We consider finite logical structures denoted $\str A,\str B,$ etc.
over a signature $\Sigma$ containing relation and/or function symbols, with the usual semantics of first-order logic, specifically, of \emph{$\Sigma$-formulas}, on such structures. A structure $\str A$ is often identified with its domain. By $|\str A|$ we denote the size of the domain of $\str A$.
We denote finite sets of variables using symbols $\tup x, \tup y,\tup z$, etc. 
If $A$ is a set and $\tup x$ is a set of variables, then
$A^{\tup x}$ is the set of all functions from $\tup x$ to $A$. Elements of $A^{\tup x}$ are called \emph{tuples}. For such a tuple $\tup a$ and a variable $x\in\tup x$, by $\tup a[x]$ we denote the value of $\tup a$ at $x$, and for a set of variables $\tup y\subset \tup x$, by $\tup a[\tup y]$ we denote the restriction of $\tup a$ to $\tup y$.  
If $\alpha$ is a first-order formula, then we may write $\alpha(\tup x)$ to indicate that the free variables of $\alpha$ are \emph{contained} in $\tup x$ ($\tup x$ may also contain additional variables).
The semantics of a first-order formula $\alpha(\tup x)$ in a structure $\str A$ is the set $\alpha_\str A\subset \str A^{\tup x}$ of tuples which satisfy $\alpha$ in $\str A$.

\subsubsection*{Classes of structures of bounded expansion}
By the Gaifman graph of a relational structure $\str A$ we mean the graph whose vertices are the elements of $\str A$, and where two distinct elements $v,w$ are adjacent if and only if there is some tuple $\tup a\in R$ containing $v$ and $w$,
for some relation $R$ in the signature of $\str A$. If $\str A$ additionally has function symbols, then its Gaifman graph is defined by considering $\str A$ as a relational structure, by replacing each function by the relation describing its graph.
\medskip

A class of structures $\CCC$ has \emph{bounded expansion} if the induced class of  Gaifman graphs  has bounded expansion. In this case, the number of tuples in a structure $\str A\in\CCC$ is linear in the size of the domain, i.e., $|R_\str A|\in \Oof_{\CCC}(|\str A|)$ for each  $R\in \Sigma$. Moreover, one can compute~\cite[Section 1.2]{Dvorak2013testing}  in time $\Oof_\CCC(\str A)$a representation of $\str A$ which allows testing membership of a given tuple in a given relation, or compute the value of a given function on a given tuple, in time $\Oof_{\CCC}(1)$. We assume this representation when working with classes of bounded expansion. Hence, for structures $\str A$ belonging to a fixed class of bounded expansion, $|\str A|$ can be regarded as the size of the representation of $\str A$.



\subsubsection*{Quantifier elimination on classes of bounded expansion}

The results of this paper build upon the following quantifier elimination result, 
which traces back to~\cite{DBLP:conf/focs/DvorakKT10,Dvorak2013testing, KazanaS13,grohe2011methods}. It immediately yields a model-checking algorithm for first-order logic with linear data complexity, for any class of bounded expansion.

\begin{theorem}{\upshape\cite[Theorem 3.7]{Dvorak2013testing}}\label{thm:qe}
    Let $\Sigma$ be a signature,
  $\phi(\tup x)$ a first-order $\Sigma$-formula,
 and $\CCC$ a class of $\Sigma$-structures of bounded expansion.
 There is a signature $\wh\Sigma$, 
 a quantifier-free $\wh\Sigma$-formula $\wh\phi(\tup x)$, a class $\wh\CCC$ of $\wh\Sigma$-structures of bounded expansion, and an  algorithm which, given a $\Sigma$-structure $\str A\in\CCC$, computes in  time $\Oof_{\CCC,\phi}(|\str A|)$ a $\wh\Sigma$-structure $\whstr A\in\wh\CCC$ with the same domain as $\str A$, such that 
$\phi_\str A=\wh\phi_\whstr A$.
\end{theorem}
Note that if $\phi$ is a sentence then $\wh\phi$ is a quantifier-free sentence
(possibly involving constants), so can be evaluated in constant time on $\whstr A$, yielding the (boolean) value of $\phi$ on $\str A$. This gives the linear-time model checking algorithm.

\section{Circuits computing weighted expressions}\label{sec:circuits}
We define structures equipped with $\sem S$-valued weight functions, where $\sem S$ is a  semiring, and define two models of computation 
for such structures: by means of expressions and by means of circuits with permanent gates. The main result of this section states that over classes of bounded expansion, a fixed expression can be compiled in linear time into a circuit of bounded depth and bounded fan-out.

\subsubsection*{Semirings}
We only consider commutative semirings.
Such a semiring $\sem S$ is a set (which we also denote $\sem S$) equipped with two commutative, associative 
operations $+$ and ${\cdot}$
with neutral elements denoted $0$ and $1$ respectively,
where 
$\cdot$ distributes over~$+$ and $0\cdot s=0$ for all $s\in\sem S$.
Examples include the boolean semiring $\sem B=(\set{\textit{true},\textit{false}},\lor,\land)$,
$(\N,+,\cdot)$, $(\N\cup\set{+\infty},\min,\max)$, $(\N\cup\set{+\infty},\min,+)$,
rings, e.g. $(\Z,+,\cdot)$, $(\Q,+,\cdot)$, boolean algebras $(P(X),\cup,\cap)$, etc.

\subsubsection*{Weighted structures and  expressions}
Fix a signature $\Sigma$ and a  finite set  of \emph{weight symbols} $\tup w=\set{w_1,w_2,\ldots},$ each with a prescribed arity $r\ge 0$.  A \emph{$\Sigma(\tup w)$-structure} over a semiring $\sem S$ is a $\Sigma$-structure $\str A$ together with 
an
interpretation of each weight symbol $w\in \tup w$ of arity $r$ as a \emph{ weight function} 
$w_\str A\from \str A^r\to \sem S$, which only assigns nonzero weights to tuples in the structure~$\str A$. This requirement means that if $r>1$ and 
 $w_\str A(\tup a)\neq 0$ then there is some relation symbol $R\in\Sigma$ of arity $r$ such that $\tup a \in R_\str A$. 
We write $\str A(\tup w_{\str A})$ to denote the above $\Sigma(\tup w)$-structure,
or simply $\str A(\tup w)$ if it does not lead to confusion.





A \emph{weighted $\Sigma(\tup w)$-expression}, or simply $\Sigma(\tup w)$-expression,
 is an expression such as for example:
$$f(z)=\sum_{x}\left(w(x)\cdot \sum_y [(x\neq y)\land \neg E(x,y)]\cdot u(x,y)+ w(y)\right).$$
Here, $[\alpha]$ denotes the \emph{Iverson bracket}, which evaluates to $1$ or $0$ depending on whether or not the first-order formula $\alpha$ holds.

Formally, weighted expressions are defined inductively, as the smallest set of expressions containing $f_1+f_2$ and $f_1\cdot f_2$, 
where $f_1,f_2$ are weighted expressions, 
 $\sum_x f$, where $x$ is a variable and $f$ and expression,
  constants $s$ for $s\in \sem S$, 
expressions $w(t^1,\ldots,t^r)$, where $w\in \tup w$ is a weight symbol of arity $r$ and $t^1,\ldots,t^r$ are terms built out of function symbols in $\Sigma$ and variables,
and expressions of the form $[\alpha]$, where $\alpha$ is a first-order $\Sigma$-formula.
We write $f(\tup x)$ to underline that the free variables of $f$ are contained in $\tup x$. A \emph{closed weighted expression} is a weighted expression with no free variables. A weighted expression 
is \emph{quantifier-free} if the subexpressions $[\alpha]$ do not involve quantifiers.

Given a $\Sigma(\tup w)$-structure $\str A(\tup w)$, the interpretation of an expression $f(\tup x)$ in $\str A$ is a function $f_{\str A(\tup w)}\from {\str A^{\tup x}}\to \sem S$, defined in the natural way. Namely, if $f(\tup x)=[\alpha(\tup x)]$ then $f_{\str A(\tup w)}(\tup a)$
is equal to $1\in \sem S$ if $\tup a\in \alpha_{\str A}$ and $0$ otherwise.
If $f(\tup x)=w(t^1,\ldots,t^r)$, then $f_{\str A(\tup w)}(\tup a)=w_\str A(t^1_\str A(\tup a),\ldots,t^r_\str A(\tup a))$.
If $f(\tup x)=g(\tup x)+h(\tup x)$ then $f_\str A(\tup a)=g_\str A(\tup a)+
h_\str A(\tup a)$, and similarly for $\cdot$ instead of $+$.
Finally, if $f(\tup x)=\sum_y g(\tup x,y)$ then $f_{\str A(\tup w)}(\tup a)=\sum_{b\in\str A} 
g_{\str A(\tup w)}(\tup a,b)$. 


\begin{example}
    Let $\phi(x,y,z)$ be a first-order formula and let $\sem S=(\mathbb Z,+,\cdot)$ be the ring of integers. Then the weighted expression 
    $\sum_{x,y,z}[\phi(x,y,z)]$ 
    counts the answers to $\phi(x,y,z)$ in a given structure $\str A$. 
    More generally,
consider the weighted query $$f=\sum_{x,y,z}[\phi(x,y,z)]\cdot p_1(x)\cdot p_2(y)\cdot p_3(y),$$
where $p_1,p_2,p_3$ are three $\mathbb Q$-valued weight functions representing probability distributions on the input structure. Then $f_\str A$ represents the probability that a random triple $(a,b,c)$ satisfies $\phi(x,y,z)$, where $a,b,c$ are selected independently with distributions $p_1,p_2,p_3$, respectively. 
By a result of Kazana and Segoufin~\cite{KazanaS13}, this value can be computed in linear time. This is generalized in Theorem~\ref{thm:algo-main}, from which it is also follows that the value $f_\str A$ can be updated in constant time, whenever a weight is modified.\qed
\end{example}




\subsubsection*{Circuits with permanent gates}
We consider circuits whose gates perform 
 operations from a semiring $\sem S$, and which input the weights of the tuples in the relations of a structure~$\str A$. Apart from the standard gates for addition and multiplication, the circuits may have \emph{permanent gates}, implementing the permanent in $\sem S$, as defined below.

 Fix a semiring $\sem S$ and 
let $M$ be a rectangular matrix with values in   $\sem S$, i.e., 
$M$ is a function $M\from R\times C\to \sem S$, where $R$ and $C$ are the sets of rows and columns of $M$, respectively.  The \emph{permanent} of $M$ is
\begin{align}\label{eq:perm}
    \perm(M)=\sum_{f}
    \prod_{r\in R}M[{r,f(r)}],        
\end{align}
where the sum ranges over all injective functions $f\from R\to C$.
Since $\sem S$ is commutative, the order of multiplication in~\eqref{eq:perm} is irrelevant, which allows us to consider matrices with unordered sets of rows and columns.
For example, if  $M$ consists of three rows, $(a_i)_{i},(b_i)_{i}, (c_i)_{i}$, then $\perm (M)=\sum_{i\neq j\neq k\neq i}a_i b_j c_k$. 
In our applications, the number of rows will be fixed (depending on the  expression), whereas the number of columns will be unbounded (depending on the data).
Note that computing the permanent of a $k\times n$ matrix amounts to evaluating the weighted expression
$$\sum_{\substack{x_1,\ldots,x_k\\\textit{distinct}}}w_1(x_1)\cdots w_k(x_k)$$ 
on the set  $\set{1,\ldots,n}$ equipped with weights $w_i(j)=M[i,j]$ for $1\le i\le k$ and $1\le j\le n$.

\medskip
A \emph{circuit with permanent gates}  is a directed acyclic graph 
whose nodes are called \emph{gates}. If there is a directed edge from gate $g$ to gate  $g'$ then $g$ is  an \emph{input} to $g'$.
The \emph{size} of a circuit is the number of edges plus the number of gates.
The \emph{fan-in} of a gate is the number of its inputs; the \emph{fan-out} is defined dually.
The \emph{reach-out} of a gate $g$ in a circuit $C$ is the number of vertices  reachable from $g$ by a directed path. The \emph{depth} of $C$
is the maximal length of a directed path.
The considered gates are of one of the following types:
\begin{itemize}[leftmargin=0.2cm]
\item    \emph{input gates} of fan-in $0$,
    \item \emph{constant gates}   of fan-in $0$, each labeled by an element of $\sem S$,
    \item \emph{multiplication gates} of fan-in two,
    \item \emph{addition gates} of arbitrary fan-in, and
    \item 
\emph{permanent gates}, whose inputs are indexed by pairs in $R\times C$, for some sets  of rows $R$ and columns $C$.   
\end{itemize}
Note that permanent gates generalize addition and multiplication gates, respectively, by considering matrices with one row and diagonal matrices.

Given a \emph{valuation} $v$ which maps the input gates to elements of a semiring $\sem S$, the value of every gate in the circuit is defined inductively in the obvious way. If a circuit $C$ has a distinguished \emph{output gate}, then by $C(v)$ we denote the value at that gate induced by the valuation $v$. Note that the same circuit $C$ can be evaluated in different semirings $\sem S$, as long as they contain the constants appearing in $C$.

It will be convenient to define circuits via \emph{terms} involving addition, multiplication, permanents, constants and variables.
The size of a term $t$ is defined as the sum over all its subterms~$t'$ (with multiple occurrences counted once) of one plus the arity of the outermost operation in $t'$ in case  $t'$ is not a variable or a constant.
For example, the term $+(\times(y,y),\times(y,y),x,x,x)$ has size $(1+5)+(1+2)+1+1=11$.
 A term of size $s$ naturally defines a circuit of size $s$.

\begin{example}\label{ex:circ}Anticipating the main result of this section,
we construct circuits evaluating  the following weighted expression: $$f=\sum_{x,y,z}[\phi]\cdot u(x)\cdot v(y)\cdot w(z),$$
    where $\phi$ is the formula $x\neq y\land x\neq z$.
    For a set $A$ with weight functions $u,v,w\from A\to \sem S$, we construct a circuit computing the value of $f$, whose inputs are the weights $u(a),v(a),w(a)$ of elements $a\in A$.
    
    Observe that $\phi$ is equivalent to the disjunction $\phi_1\lor \phi_2$ 
    where $\phi_1\equiv x\neq y\land x\neq z\land y\neq z$ and $\phi_2\equiv x\neq y\land y=z$. Moreover, the disjunction $\phi_1\lor \phi_2$ is mutually exclusive, i.e., $\phi_1\land \phi_2$ is not satisfiable. It follows that $f$ is equivalent to the expression $f_1+f_2$, where $f_1$ and $f_2$ are obtained from $f$ by replacing $\phi$ by $\phi_1$ and $\phi_2$, respectively.
    We construct circuits $C_1$ and $C_2$, computing the values of $f_1$ and $f_2$, respectively. 
    The circuit $C$ computing the value of $f$ is then constructed by creating an addition gate whose inputs are the outputs of $C_1$ and of $C_2$.
    
    Note that the value of $f_1$ is just the permanent of the $3\times A$ matrix  where column $a\in A$ has entries $u(a), v(a), w(a).$ This immediately yields the circuit $C_1$ which has one permanent gate, applied directly to the input gates, evaluating  $f_1$ on $A$.
    For $f_2$, observe that it is equivalent to $\sum_{x,y}[x\neq y]\cdot u(x)\cdot (v(y)\cdot w(y))$. The circuit $C_2$ computing $f_2$ on $A$ is therefore a $2\times A$ permanent gate, where column $a$ has entries $u(a)$ and $v(a)\cdot w(a)$, the latter being a multiplication gate applied to two input gates.
    \qed
\end{example}

\medskip We now generalize the idea from Example~\ref{ex:circ} 
and define how circuits can evaluate weighted expressions in arbitrary structures. 
Fix a signature $\Sigma$ and a set of weight symbols $\tup w$.
We consider $\sem S$-circuits as recognizers of $\Sigma(\tup w)$-structures, as follows. 
A \emph{$\Sigma(\tup w)$-circuit over a structure $\str A$} is an $\sem S$-circuit $C_\str A$ whose inputs are pairs $(w,\tup a)$,
where $w\in \tup w$ is a weight symbol of arity $r$ and $\tup a$ 
is an $r$-tuple of elements which belongs to some relation $R_\str A$ if $r>1$.
A $\Sigma(\tup w)$-structure $\str A(\tup w_{\str A})$  
 determines a valuation of the input gates of $C_\str A$ in $\sem S$,
where the value assigned to the gate $(w, \tup a)$
is $w_{\str A}(\tup a)$. 
We write $C_{\str A}(\tup w_\str A)$ to denote the output of the circuit $C_\str A$ induced by this valuation. Hence,  $\str A$ is considered to be fixed, whereas the weights $\tup w_\str A$ may vary, and are the inputs to the circuit.

Let $\CCC$ be a class of $\Sigma$-structures and let 
$(C_{\str A})_{\str A\in \CCC}$ be a circuit family such that $C_\str A$ is a $\Sigma(\tup w)$-circuit over $\str A$,
and let $f$ be a closed $\Sigma(\tup w)$-expression.
We say that $(C_{\str A})_{\str A\in \CCC}$
\emph{computes  $f$} 
if for every 
$\Sigma$-structure $\str A$ with $\str A\in\CCC$, every semiring $\sem S$,
and weight functions $w_\str A\from\str A\to\sem S$ for $w\in\tup w$,
$$C_\str A(\tup w_\str A)=f_{\str A(\tup w_\str A)}.$$
The family $(C_{\str A})_{\str A\in \CCC}$ is \emph{linear-time computable} if there is an algorithm which given $\str A\in\CCC$ outputs $C_{\str A}$ in time $\Oof_{f,\CCC}(|\str A|)$. 



\subsubsection*{Circuits for sparse structures}
We are ready to state the main result of this section. Here and later (\appmark) indicates that the proof is in the appendix.
\begin{theorem}[\appmark]\label{thm:circuits}
    Fix a signature $\Sigma$, a set of weight symbols $\tup w$, and a closed $\Sigma(\tup w)$-expression $f$.
 Let $\CCC$ be a class of $\Sigma$-structures of bounded expansion. There is a linear-time computable family $(C_\str A)_{\str A\in \CCC}$ of $\Sigma(\tup w)$-circuits with permanent gates which computes $f$ and has 
        bounded depth, 
        bounded fan-out, and 
         bounded number of rows in the permanent gates.        
    The constants implicit in the above statements depend only on the expression $f$ and on the class $\CCC$.

\end{theorem}
\newcommand{\AC}{\mathsf{AC}}
\newcommand{\paraAC}[1]{\textrm{para-}\AC^{#1}}

Note that the same circuit $C_\str A$ can be used to evaluate $f$ in an arbitrary semiring $\sem S$. This universal property captures a strong structural property of the circuits $C_\str A$, which will be then exploited for enumerating answers to first-order queries without repetitions.

Theorem~\ref{thm:circuits} is the cornerstone of our framework, 
and will allow us to reduce various problems concerning 
classes of bounded expansion to an algebraic analysis of the permanent over a suitably chosen semiring. This will be illustrated in the following sections.




\begin{example}\label{ex:forests}
Let $\Sigma$ consist of a unary function $f$ and a unary predicate $R$, and let $\CCC$ be the class of forests of depth $1$, treated as $\Sigma$-structures, 
where in given a forest, $R$ marks all the roots
and $f$ maps each non-root node to its parent and each root to itself.
Consider the weighted expression 
\begin{align}\label{eq:ff}
f=\sum_{x,y}[\phi]\cdot u(x)\cdot v(y),
\end{align}
where $\phi$ is a quantifier-free $\Sigma$-formula.
Given a forest $F$, we construct a circuit $C_F$ evaluating $f_F$ on $F$, as obtained from Theorem~\ref{thm:circuits}.

The formula $\phi$ can be written as a mutually exclusive disjunction of \emph{quantifier-free types}, fully specifying the relationship between $x$ and $y$
(up to equivalence in $\CCC$).
Examples of such quantifier-free types are\todo{pics}:
\begin{align*}
    \phi_1&\equiv R(x)\land R(y) \land x\neq y, \\
    \phi_2&\equiv  \neg R(x)\land \neg R(y)\land x\neq y\land f(x)=f(y), \\
    \phi_3&\equiv \neg R(x)\land \neg R(y)\land  f(x)\neq f(y).    
\end{align*}
Then $f$ is equivalent to the sum of expressions obtained from $f$ by substituting $\phi$ for the quantifier-free types occurring in the disjunction.
Hence, it suffices to consider the case when $\phi$ is a quantifier-free type, as we can then add up the outputs of the circuits obtained for such expressions $f$. 

For example, in the case when $\phi=\phi_2$, note that
$$f_{F(\tup w)}=\sum_{r\in R}\sum_{\substack{a,b\in f^{-1}(r)\\ a\neq b}}u(a)\cdot v(b),$$
which readily yields a circuit: a summation gate of fan-in $|R|$, 
whose input corresponding to $r\in R$ is a permanent gate applied to a $2\times f^{-1}(r)$-matrix of input gates.

When $\phi=\phi_3$, note that
$$f_{F(\tup w)}=\sum_{\substack{r,s\in R\\r\neq s}}\left(\sum_{a\in f^{-1}(r)}u(a)\right)\cdot \left(\sum_{b\in f^{-1}(r)}v(b)\right),
$$
which yields a circuit with a $2\times R$-permanent gate, whose inputs are summation gates.\qed

\end{example}

\begin{proof}[\ifdoublecolumn{Proof} sketch for Theorem~\ref{thm:circuits}]
The proof  proceeds as follows.
By Theorem~\ref{thm:qe}, it is enough  to consider only the case when $f$ is  quantifier-free,
i.e., its subexpressions $[\alpha]$ do not involve quantifiers.
For such expressions $f$,
we prove 
 the result in special cases of increasing generality:

    \noindent
    \textbf{Case 1:} \emph{ $\CCC$ is a class of vertex-labelled forests of bounded depth}, where each forest is treated as a structure with unary predicates for representing the labelling, and the  function mapping each non-root node to its parent. 
    This case generalizes Example~\ref{ex:forests}, and proceeds by induction on the depth. 
    
    \noindent\textbf{Case 2:} \emph{$\CCC$ is a class of structures of bounded treedepth, over a relational signature with binary relations only}. This case reduces to the previous 
    case, by encoding structures of bounded treedepth as labelled forests of bounded depth, similarly as in Example~\ref{ex:homo}.


    \noindent\textbf{Case 3:} \emph{$\CCC$ is a class of structures over a  relational signature with binary relations only}. This case reduces to the previous case, using the fact that the graphs admit low treedepth colorings,
    similarly as in Example~\ref{ex:homo}.
    
    \noindent\textbf{Case 4:} $\CCC$ is a class of structures of bounded expansion. This case reduces to the previous one,   using an 
    encoding of structures from a bounded expansion class by colored graphs.
\end{proof}





\section{Weighted query evaluation}\label{sec:eval}
Thanks to Theorem~\ref{thm:circuits},
 many algorithmic problems concerning the evaluation of an expression $f$ on a class of $\Sigma(\tup w)$-structures  boil down to analogous problems for permanents. 
For example, if the permanent of a $k\times n$ matrix over a semiring $\sem S$ 
can be computed in  time $\Oof_k(n)$, then the same holds for the evaluation of any $\Sigma(\tup w)$-expression $f$ over any $\Sigma$-structure $\str A$ from a fixed class with bounded expansion. 
This turns out to be the case for any semiring $\sem S$, assuming the unit cost model.
Similarly, if there is a dynamic algorithm which maintains the  value of the permanent of a matrix $M$ in constant time
upon single-entry modifications to $M$,
then there is a dynamic algorithm which maintains in constant time the value of $f_{\str A(\tup w)}$ upon  updates to the weights $\tup w$. This turns out to be the case for example if $\sem S$ is a ring or is finite.

In this section, we study how permanents can be efficiently computed and updated in various semirings. This will yield the following result, in which we consider updates to single weights.

\begin{theorem}\label{thm:algo-main}
    Fix a $\Sigma(\tup w)$-expression $f(\tup x)$, a semiring $\sem S$ whose operations take constant time, and a class $\CCC$ of $\Sigma$-structures of bounded expansion. There is an algorithm which, given a 
    $\Sigma(\tup w)$-structure $\str A(\tup w)$ with $\str A\in\CCC$,
    computes in linear time a dynamic data structure allowing to query the value $f_{\str A(\tup w)}(\tup a)$ at a given tuple $\tup a\in \str A^{\tup x}$ in 
    logarithmic time, and to maintain  updates to the weights $\tup w$ in logarithmic time.
 If $\sem S$ is a ring or a finite semiring, then querying and updating  is done in constant time.
\end{theorem}

\begin{example}\label{ex:pagerank}
    Consider the following weighted query\footnote{We thank an anonymous reviewer for suggesting this query.}, evaluated in the ring of rationals, which computes the weights in a subsequent round 
    of the PageRank algorithm on a directed graph $G=(V,E)$, given the weights from the previous round. 
    $$f(x)=\frac{1-d}N+d\cdot \sum_y [E(y,x)]\cdot \frac {w(y)}{l(y)}.$$
    Here, $w(a)$ denotes the weight in the previous round, and $l(a)$ is the out-degree of $a$, $N$ is the total number of pages and $d\le 1$ is a damping factor.
    Formally, as we don't allow division, we represent $1/l(\cdot)$ as a weight function.  Theorem~\ref{thm:algo-main} yields an algorithm which computes in linear time a dynamic data structure which allows to query the value $f(a)$ in constant time, and maintains updates  to weights in constant time. \qed
\end{example}

To prove the theorem, we analyze the permanent over various semirings. Throughout Section~\ref{sec:eval}, we assume that the semiring operations of $\sem S$ take constant time. 

\subsubsection*{Arbitrary semirings}To prove Theorem~\ref{thm:algo-main},
we consider first the case when $f$ has no free variables, as the general case reduces to this case.
We show that there is an algorithm computing the value $f_{\str A(\tup w)}$ in linear time.
Additionally, the value can be maintained in logarithmic time, upon updates to $\tup w$.
We start by showing this for the permanent gates. In fact, we show that permanent gates can be replaced by circuits in which every gate can reach only a logarithmic number of other gates. This allows to update the output of the circuit in logarithmic time, when a single input is modified.

We remark that the evaluation of permanents of rectangular matrices in commutative semirings has been studied before, see~\cite{Bjorklund:2010:EPR:1840004.1840168} for an overview. Our algorithm below allows computing the permanent of a $k\times n$ matrix in time $2^{\Oof(k\log k)}\cdot n$.
The dependency on $k$ can be decreased~\cite{Bjorklund:2010:EPR:1840004.1840168} to $\Oof(k\cdot 2^k)$. As $k$ depends on the query in our case, and we are not concerned with query complexity in this paper, we do not explore such optimizations.

\smallskip
First, we define a variant of the permanent which is easier to analyze here.
Given a matrix $M$ with rows $R$ and columns $C$, where $R$ and $C$ are both totally ordered, define $\perm'(M)$ by the same expression as the permanent, but where $f$ only ranges over increasing functions $f\from R\to C$. Fix an ordering of $C$.
Then, $\perm(M)=\sum_{<_R}\perm'(M)$, where $<_R$ ranges over all $|R|!$ orderings of $R$. Since $|R|$ is considered fixed, it suffices to 
produce circuits computing $\perm'(M)$.

Suppose $M$ is a $k\times n$ matrix,
and let $1\le \ell\le n$. 
By $A_i^\ell$ denote 
 the submatrix of $M$ with rows $\set{1,\ldots,i}$ and columns 
$\set{1,\ldots,\ell}$,
and by $B_i^\ell$ denote the submatrix with rows $\set{i+1,\ldots,k}$ and columns $\set{\ell+1,\ldots,n}$. 
The following lemma follows by grouping 
increasing functions $f\from [k]\to[n]$ according to the moment they exceed $\ell$, i.e.,
$\max\setof{0\le i\le k}{f(i)\le  \ell}$.
\begin{lemma}\label{lem:recursion}
For a $k\times n$ matrix $M$ and $1\le \ell\le n$,
the  following identity holds:
\begin{align}\label{eq:rec}
    \perm'(M)=&
    \sum_{i=0}^k\perm'(A_i^\ell)\cdot \perm'(B_i^\ell)
\end{align}
\end{lemma}

Note that if $\ell=\lfloor n/2\rfloor$, all the matrices involved in the sum have roughly half the number of columns as $M$. 
Recursively expanding the permanents on the right-hand side, we obtain a term 
which -- viewed as a circuit -- has the  properties expressed in the  following lemma.

\begin{lemma}\label{lem:arb-semi}
    Fix a semiring  $\sem S$ and  a number $k\in\N$.
   Then there is a family $(C_{n})_{n\in \N}$ of circuits without permanent gates computing the permanent of $k\times n$ matrices, and with 
        depth $\Oof_k(\log n)$, fan-out $\Oof_k(1)$,
        reach-out $\Oof_k(\log n)$,
   and where each circuit $C_n$ is computable  in time $\Oof_k(n)$.
\end{lemma}

This immediately yields the following.
\begin{corollary}
    The circuits $C_\str A$ in Theorem~\ref{thm:circuits} can be assumed to have no permanent gates, but logarithmic reach-out instead of bounded depth.
\end{corollary}
Note that a circuit $C$ with reach-out $c$  yields a dynamic data structure which allows maintaining the value $C(v)$ under updates modifying the valuation $v$ 
in time $\Oof(c)$. This yields the following:
\begin{corollary}\label{cor:arb-sem}
    There is an algorithm which, given a 
    weighted structure  $\str A(\tup w)$ with $\str A\in \CCC$, computes in time  $\Oof_{\CCC,f}(|\str A|)$
    a dynamic data structure maintaining the value $f_{\str A(\tup w)}$ which 
    allows to update  $\tup w$ 
    in  time $\Oof_{\CCC,f}(\log|\str A|)$.
\end{corollary}


Note that for the semirings $(\N\cup\set {+\infty},\min,\max)$ or 
$(\N\cup\set {+\infty},\min,+)$, the permanent of a $1\times n$ matrix is just the smallest entry of the matrix, and maintaining the smallest entry allows to implement sorting in the comparison model. This leads to the following.

\begin{proposition}\label{prop:sorting}
The logarithmic update time in Corollary~\ref{cor:arb-sem} and Theorem~\ref{thm:algo-main} is tight, i.e., it cannot be improved for general semirings $\sem S$ while keeping the initialization time linear.
\end{proposition}


Therefore, to obtain a dynamic algorithm with linear initialization time and  sublogarithmic update time, we need to assume some additional properties of the semiring $\sem S$.

\subsubsection*{Rings}
The next  lemma states that for rings, permanent gates can be completely eliminated.

\begin{lemma}[\appmark]\label{lem:perm-rings}
    Suppose that  $\sem S$ is a ring and fix a number $k\in\N$.
   Then there is a family $(C_{n})_{n\in \N}$ of $\sem S$-circuits computing the permanent of $k\times n$ matrices, which has
        bounded depth,
        bounded fan-out,
   and where each circuit $C_n$ is computable  in time $\Oof_k(n)$.
\end{lemma}
For example, for $k=2$, the permanent of a matrix  with two rows $(a_1 \ldots a_n)$ and $(b_1 \ldots b_n)$ can be expressed as follows:
    $$\sum_{i\neq j}a_ib_j=\sum_i a_i\cdot \sum_i b_i-\sum_i a_i\cdot b_i,$$
    and the right-hand side readily yields a term using only addition and multiplication, possibly by $-1$. In general, we use an inclusion-exclusion formula.

\begin{corollary}
    If in Theorem~\ref{thm:circuits} we consider rings only, then the obtained circuits $C_\str A$ have no permanent gates. 
\end{corollary}
Note that a circuit $C$ with fan-out $c$ and constant depth $d$ has reach-out $\Oof(c^{d+1})$, yielding constant-time updates:
\begin{corollary}\label{cor:ring}
    There is an algorithm which, given a 
    weighted structure $\str A(\tup w)$ with $\str A\in\CCC$ and weight functions $\tup w$ taking values in a ring $\sem S$,  computes in linear time 
    a dynamic data structure maintaining the value $f_{\str A(\tup w)}$ which 
    allows to update  $\tup w$ 
    in constant time.
\end{corollary}

\subsubsection*{Finite semirings}
Another case where we can replace permanent gates by simpler circuits is the case of finite semirings. For those semirings, we   replace permanent gates by circuits with \emph{counting} gates: \emph{threshold} gates and \emph{mod} gates (the precise definition is in the appendix).
The key  observation is that the permanent of a $k\times n$ matrix $M$ can be computed  based on the 
    number of occurrences of each tuple $c\in \sem S^k$ as a column in $M$. 
    This leads to the following lemma.

\begin{lemma}[\appmark]\label{lem:counting}
    Fix a finite semiring $\sem S$ and a number $k\in\N$.
    There is a family $(C_{n})_{n\in \N}$ of  $\sem S$-circuits with counting gates computing the permanent of $k\times n$ matrices, which has
         bounded depth,
         bounded fan-out,
    and each circuit $C_n$ is computable  in time~$\Oof_k(n)$, has thresholds bounded by $\Oof(|\sem S|+k)$, and the modulus of each mod gate is the order of a cyclic subgroup of $(\sem S,+)$.
\end{lemma}

\begin{corollary}\label{cor:finite-semirings}
    If $\sem S$ is a finite semiring, then the circuits $C_\str A$ in Theorem~\ref{thm:circuits} can be assumed to be  $\sem S$-circuits with counting gates with thresholds at most $\Oof_f(|\sem S|)$ and modulus as above. In particular, for the boolean semiring $\sem B$, the mod gates are absent.
\end{corollary}
If we assume a computation model which allows incrementing and decrementing   integers, as well as testing 
divisibility by any fixed number,
 in constant time, then we get the following.

\begin{corollary}\label{cor:fin-sem}
    There is an algorithm which, given 
    a 
    weighted structure $\str A(\tup w)$ with $\str A\in  \CCC$, computes in linear time 
    a dynamic data structure maintaining the value $f_{\str A(\tup w)}$ which
    allows to update  $\tup w$ 
    in constant time.
\end{corollary}

\begin{proof}[\ifdoublecolumn{Proof} of Theorem~\ref{thm:algo-main}]
 The case when $f$ has no free variables  follows from Corollaries~\ref{cor:arb-sem}, \ref{cor:ring}, and \ref{cor:fin-sem}. The case of an expression $f(\tup x)$ with 
 free variables $\tup x=\set{x_1,\ldots,x_k}$ reduces to that case, by considering
  the closed expression $$f'=\sum_\tup x f(\tup x)\cdot v_1(x_1)\cdots v_k(x_k),$$ where 
 $v_1,\ldots,v_k$ are new unary weight symbols, set to $0$ by default. Then $f_{\str A(\tup w)}(\tup a)=f'_{\str A(\tup w\tup v)}$ where for $i=1,\ldots,k$, the function $v_i$ maps $a_i$ to $1$ and all other elements to $0$.
Hence, querying the value of $f$ at a tuple $\tup a$ can be simulated by $2|\tup x|$ updates, as it amounts to temporarily setting the weights $v_1(a_1),\ldots,v_k(a_k)$ to $1$, querying the value of $f'$, and then setting the weights back to $0$.
\end{proof}


We remark that for some semirings, simulating querying 
by a sequence of updates may not be optimal.
For example, if $\sem S=(\mathbb N\cup \set{+\infty},\min,+)$, then
querying can be achieved in constant time,
while still having logarithmic time updates and linear time initialization,
even though updating provably requires logarithmic time, due to Proposition~\ref{prop:sorting}. To achieve these improved query times, again it is sufficient to solve the problem for the case of a permanent, where querying amounts to computing the permanent of the current matrix with a fixed number of values temporarily modified. For semirings such as above,
constant-time querying can be achieved using heaps. We omit the details.



\section{Provenance semiring}

Provenance analysis is about tracing
how a given query output was produced, and what are the properties of the computation which lead to this output.
For example, one could try to analyze which input tuples are responsible for the fact that a given tuple is produced in the output, or what is the probability that a given output is produced, basing on some probability distribution on the inputs. The algebraic theory of provenance is based on semirings. Indeed, many instances of provenance analysis can be seen as evaluating queries in 
fixed semirings. It is known that   the \emph{free semiring}, considered below, is in a certain precise sense the most general semiring which can be used for provenance analysis~\cite{provenance}.




\paragraph{The free semiring.}
\newcommand{\eol}{\bot}

Let $A$ be a (possibly infinite) set of symbols. The \emph{free (commutative) semiring} generated by $A$, denoted $\sem F_A$,
is the  semiring consisting of sums of unordered sequences of elements of $A$,
with addition and multiplication defined naturally.
The $0$ of this semiring is the empty sum, 
and the $1$ is the sum with one summand, being the empty sequence. This semiring is isomorphic to the semiring of polynomials with variables from $A$ and coefficients from $\N$. It is also called the \emph{provenance semiring}~\cite{provenance}.

\begin{example}Let $\Sigma$ be a signature consisting of a binary relation symbol $E$, and consider the query $$\phi(x)=\exists_{y,z} E(x,y)\land E(y,z)\land E(z,x).$$
    For the purpose of tracking  provenance,
consider the expression $$f(x)=
\sum_{y,z}  w(x,y) \cdot w(y,z)\cdot w(z,x).$$ 
Suppose  $G=(V,E)$ is the directed graph with vertices $a,b,c,d$ and edges $ab, bc, ca, bd,da$.
Then, $\phi_G(a)$ holds,
and evaluating $f(a)$ in $G(w)$,
where $w_G(x,y)=e_{xy}$ is a unique identifier for each $(x,y)\in E$,
reveals  \emph{why} it holds:  $$f_\str A(a)=
e_{ab}e_{bc}e_{ca}+e_{ab}e_{bd}e_{da}.
$$
\end{example}
\smallskip

\paragraph{Iterators.}
As the elements of the free semiring be sums of a number of summands which depends on the data, 
 we  cannot reasonably assume that such elements can be represented in single memory cells. We therefore assume that each element of the free semiring is represented by an iterator, as follows.

A \emph{bi-directional  iterator} for a list $u_1,\ldots,u_\ell$
 stores an index $i\in\set{0,1,\ldots,\ell}$, initially set to $1$, and implements the following methods:
\begin{description}
    \item[current:] output $u_i$, or $\eol$ if $i=0$.
    \item[next:] increment $i$, modulo $\ell$,
    \item[previous:] decrement $i$, modulo $\ell$,    
\end{description}
We say that 
 such a bi-directional iterator 
 has \emph{access time} $t$, where $t\in\N$,
 if the operations \textbf{next},
 \textbf{previous} take time $t$,
 and the operation \textbf{current}, takes time $t\cdot |u_i|$, where $|u_i|$ is the length of $u_i$.
 
 In the theorem below, for a given $\Sigma(\tup w)$-structure $\str A$, $\tup w_\str A$ is a valuation in the semiring $\sem F_A$, given by providing a bi-directional iterator representing $w_{\str A}(a)$, for each $a\in \str A$ and $w\in \tup w$. An \emph{update} to $\tup w$ modifies a single value $w_{\str A}(a)$, by providing an iterator for the new value.

 \begin{theorem}[\appmark]\label{thm:provenance}
     Let $\CCC$ be a class of $\Sigma$-structures of bounded expansion and let $f(\tup x)$ be a  $\Sigma(\tup w)$-expression.
     There is an algorithm which, given a structure $\str A\in \CCC$ and a tuple $\tup w$ of weight functions with values in $\sem F_A$, where for each $w\in \tup w$ and $a\in \str A$,
     the element $w_{\str A}(a)$ is represented by a bi-directional iterator with access time bounded by a constant, computes in linear time a data structure  
     which allows to query any tuple $\tup a\in \str A^{\tup x}$ and obtain in constant time
     a bi-directional iterator for $f_{\str A(\tup w)}(\tup a)$ with constant access time.
     The data structure is maintained in constant time upon updates to $\tup w$.
 \end{theorem}
Note that the element  $f_{\str A(\tup w)}(\tup a)$ of $\sem F_A$ may have repeating summands, and the enumerator will enumerate each repetition separately. This will not be a problem when we apply the result above in Theorem~\ref{thm:enum} later to enumerate all answers to a first-order query
 \emph{without} repetitions.

 To prove Theorem~\ref{thm:provenance}, we apply our circuit framework and Theorem~\ref{thm:circuits} in particular, which essentially tell us  that we only need to consider the case of permanents, handled below:
\begin{lemma}[\appmark]\label{lem:perm-iter}
    There is a dynamic algorithm which, given an $R\times C$ matrix $M$ with entries from $\sem F_A$, where  each entry $M[r,c]$ is represented by a bi-directional iterator with access time $t$,
computes in  time $\Oof_{R}(C)$ a data structure which maintains   a 
bi-directional iterator for $\perm(M)$ with access time $\Oof_R(t)$,
and updates it in constant time, whenever an entry of $M$ is updated.
\end{lemma}

 \section{Dynamic query enumeration}\label{sec:enum}
 We now apply our framework to yield a dynamic algorithm for enumerating the answers $\phi_\str A$ to a first-order query $\phi(\tup x)$. Below, we consider updates to a $\str A$ structure which can add/remove single tuples to its relations provided that they preserve the Gaifman graph. This amounts to saying that a tuple $\tup a$ can be added to a relation only if its elements form a clique in the Gaifman graph of $\str A$.

\begin{theorem}[\appmark]\label{thm:enum}
    Let $\Sigma$ be a signature,
 let $\phi(\tup x)$ be a first-order $\Sigma$-formula,
 and let $\CCC$ be a class of $\Sigma$-structures of bounded expansion.
 There is an algorithm which 
 inputs a structure $\str A\in\CCC$ and computes in linear time a dynamic data structure which provides a constant access, bi-directional iterator for $\phi_{\str A}$.
 The data structure is maintained  in constant time upon  updates to $\str A$ which preserve the Gaifman graph.
\end{theorem}
Note that the static version of Theorem~\ref{thm:enum}, without the last part, is the result of Kazana and Segoufin~\cite{KazanaS13}.



\begin{proof}[\ifdoublecolumn{Proof }sketch for Theorem~\ref{thm:enum}]
    We sketch the proof in the static case, thus reproving the main result of~\cite{KazanaS13}. 
    
    By Theorem~\ref{thm:qe} it is enough to consider the case when $\phi(\tup x)$ is quantifier-free. 
Fix an input structure $\str A$ and an enumeration  $\tup x=\set{x_1,\ldots,x_k}$. 
 Let $\sem F$ be the free semiring generated by 
identifiers of the form $e_a^i$, where $i\in\set{1,\ldots,k}$ and $a\in\str A$.
Its elements are expressions such as 
$e_a^1 e_b^2 e_c^3+e_b^1 e_c^2 e_d^3.$
For each $i$ with $1\le i\le k$, define an $\sem F$-valued weighted function $w_i$ on $\str A$, where $w_i(a)=e_a^i\in \sem F$ for $a\in\str A$. 
Note that for $a_1,\ldots,a_k\in\str A$, the expression $w_1(a_1)\cdots w_k(a_k)$ evaluates in the structure $\str A(\tup w)$ 
to the  element $e_{a_1}^1\cdots e_{a_k}^k$ of $\sem F$, representing the tuple $\tup a\in\str A^{\tup x}$ with $\tup a[x_i]=a_i$.
Consider the $\Sigma(\tup w)$-expression:
\begin{align}\label{eq:rrr}
f=\sum_{\tup x}[\phi(\tup x)]\cdot w_1(x_1)\cdots w_k(x_k).
\end{align}
Then $f_{\str A(\tup w)}$  is  the element of $\sem F$  representing the set of tuples $\phi_\str A\subset \str A^{\tup x}$,
with one occurrence per each tuple.
Note that each weight $w_i(a)$ is a single element $e^i_a\in \sem F$, so can be trivially represented by a bi-directional iterator with constant access time.
Applying Theorem~\ref{thm:provenance}
yields a bi-directional iterator with constant access time, which iterates through all the 
summands of the element~$f_{\str A(\tup w)}$, corresponding to the tuples in $\phi_\str A$.

The dynamic case proceeds similarly, but instead of Theorem~\ref{thm:qe}, uses its dynamic version~\cite[Theorem 6.3]{Dvorak2013testing} to reduce to the case where $\phi(\tup x)$ is quantifier-free. We then maintain the representation of each relation symbol $R$ and its negation as a weight function with values in $\set{0,1}\subset \sem F$, and replace the expression $[\phi(\tup x)]$ in~\eqref{eq:rrr} by an equivalent expression involving those weight functions. Finally, we  utilize the dynamic data structure given by Theorem~\ref{thm:provenance} to maintain an enumerator for $f_\str A(\tup w)$. The details are in the appendix.
\end{proof}

Theorem~\ref{thm:enum} only allows updates which preserve the Gaifman graph of the input structure. This includes, in particular, updates which only modify unary predicates.
The following example\footnote{suggested by 
Zden\v ek Dvo\v r\'ak and also Felix Reidl, Michał Pilipczuk, Sebastian Siebertz} shows that already this 
    yields a useful data structure.
\begin{example}
    A \emph{local search algorithm} searches for an optimal solution -- e.g. largest independent set or smallest dominating set in a given graph -- by iteratively trying to improve the  the current solution
    \emph{locally}, i.e., by modifying in each round the current solution by adding/removing at most $\lambda$ vertices, for some fixed \emph{locality radius} $\lambda$. In case of independent set or dominating set, the current solution can be represented by a unary predicate, 
    and a fixed first-order formula (depending on $\lambda$) can determine the existence of a local improvement.
    Moreover, using enumeration, we can find any such improvement in constant time, and then use it to improve the current solution by performing a constant number of updates, in constant time. Hence, 
    for graphs from a fixed class of bounded expansion, 
    using the data structure provided by Theorem~\ref{thm:enum},    
    each round of the local search algorithm can be performed in constant time, and the local search algorithm will compute a local optimum in linear time.
    
    This observation can be combined with a recent result of Har-Peled and Quanrud~\cite{Har-PeledQ17}, stating that local search with sufficiently large radius $\lambda$ yields a $(1+\eps)$-approximation 
    for these problems on any fixed class of graphs with \emph{polynomial expansion}, where $\lambda$ depends only on $\eps$, $r$, and the class. The definition of a polynomial expansion class is obtained by requiring that the function $r\mapsto c_r$ in the definition of a bounded expansion class is polynomial. These classes include e.g. the class of planar graphs.
    This  yields efficient \emph{linear-time} approximation schemes (in particular, EPTASes) for the {\sc{Distance-$r$ Independent Set}} and {\sc{Distance-$r$ Dominating Set}} problems on any class of
    graphs with polynomial expansion. This improves the result of~\cite{Har-PeledQ17}, which only obtains a PTAS.
\end{example}


\section{Nested weighted query evaluation}
We now introduce \emph{nested weighted queries} which can handle multiple semirings.
They involve a slightly different syntax than the one used previously.

If $A$ is a set and $\sem S$ is a semiring, then an \emph{$\sem S$-relation} of arity $k$
on $A$ is a function $R\from A^k\to \sem S$.
We consider structures $\str A$ equipped with relations
with values in various semirings.
Classical relational structures are recovered by only allowing relations with values in the boolean semiring.
Formulas can be constructed using
semiring summation $\sum$
playing the role of quantification,
and semiring operations $+$ and $\cdot$ playing the role
of connectives. Additionally
we allow other connectives which can
 transfer between semirings.
This is
defined below.

\paragraph{Connectives.}
Let $\ops$ be a collection
of semirings and of
\emph{connectives} which are functions  $c\from \sem S_1\times\cdots\times\sem S_k  \to \sem S$, where $\sem S,\sem S_1,\ldots,\sem S_k\in \ops$ are semirings.
Examples of connectives include
 the
 function $<\from\sem N\times\sem N\to\sem B$,
 corresponding to the total order on $\sem N$, or the function $/\from\sem Q\times \sem Q\to \sem Q$, where $\sem Q$ is the field of rationals, which maps a pair $p,q$ to $\frac p q$ if $q\neq 0$ and to $0$ otherwise.
 Another example is the
 \emph{Iverson bracket}
 denoted $[\,\cdot\, ]_{\sem S}\from \sem B\to \sem S$, where $\sem S$ is a semiring, mapping $0\in \sem B$ to $0\in \sem S$ and $1\in\sem B$ to $1\in \sem S$.

\paragraph{Signatures, structures, and Gaifman graphs.}
A \emph{$\ops$-signature} $\Sigma$ is a family
of \emph{$\sem S$-relation symbols},
for each semiring $\sem S\in\ops$,
and function symbols,
each with a specified arity.
A \emph{$\Sigma$-structure} $\str A$ is a set $A$
together with a function $f_{\str A}\from A^k\to A$ for
each function symbol $f\in \Sigma$ of arity $k$,
and together with an $\sem S$-relation
$R_{\str A}\from A^k\to \sem S$ for each $\sem S$-relation symbol $R\in \Sigma$ of arity $k$.
If $\Sigma$ contains no function symbols, then the \emph{Gaifman graph} of a $\Sigma$-structure $\str A$  is the graph whose vertices are the elements of $\str A$,
and where two distinct elements $v,w$ are adjacent
if there is some $R\in\Sigma$ and some
tuple $\tup a\in R_\str A$ such that $R(\tup a)\neq 0$ and $\tup a$ contains $v$ and $w$. As usual, we may convert a structure with function symbols into a structure which only uses relation symbols,
by interpreting a function $f\from \str A^k\to \str A$ as a $\sem B$-relation $R\from \str A^{k+1}\to\sem B$ representing the graph of $f$. Via this conversion, we define the Gaifman graph of a structure with function symbols. A class of $\Sigma$-structures has bounded expansion if the class of its Gaifman graphs has this property.

\paragraph{The logic $\fops$}
Fix a set $\ops$ of semirings and connectives, and a $\ops$-signature $\Sigma$.
We implicitly assume that $\ops$ contains the boolean semiring $\sem B$ and 
that $\Sigma$ contains the binary equality symbol $=$ as a binary $\sem B$-relation symbol.

We define the syntax of the logic $\fops$.
Each formula $\phi$
has a specified \emph{output type},
which is a semiring $\sem S\in\ops$.
If the output type is $\sem S$, then we say that  $\phi$ is an \emph{$\sem S$-valued formula}.
 The set of $\sem S$-valued $\fops$-formulas 
is defined inductively, and consists of the following formulas:

\begin{itemize}[leftmargin=0.4cm]
    \item
    $R(t^1,\ldots,t^k)$, where  $R\in \Sigma$ is an $\sem S$-relation symbol of arity $k$
    and $t^1,\ldots,t^k$ are terms built out of
variables and function symbols from $\Sigma$.

    \item
    $s$, where  $s\in\sem S$ is treated as a constant.

    \item
    $c(\phi^1,\ldots,\phi^k)$, where  $\phi^1,\ldots,\phi^k$ are $\fops$-formulas with output types $\sem S_1,\ldots,\sem S_k$, respectively,
     and $c\from \sem S_1\times \cdots \times \sem S_k\to\sem S$ is a connective in $\ops$.

    \item
     $\sum_x \phi$, where $\phi$ is an $\sem S$-valued $\fops$-formula.

    \item
    $\phi_1+\phi_2$ and $\phi_1\cdot \phi_2$, where $\phi_1$ and $\phi_2$ are $\sem S$-valued $\fops$-formulas.
    
    \item
     $[\phi]_\sem S$, where $\phi$ is a $\sem B$-valued $\fops$-formula.

    \item
    $\neg \phi$, if $\sem S=\sem B$ and $\phi$ is a $\sem B$-valued $\fops$-formula.
\end{itemize}



For a $\sem B$-valued formula $\phi$ in $\fops$
we write
$\lor$ and $\land$ instead of $+$ and $\cdot$, respectively,
$\exists_x\phi$ instead of $\sum_x\phi$,  
and
 $\forall_x\phi$ as syntactic sugar for $\neg \exists_x\neg \phi$.
Note that
 first-order logic coincides with
 $\fops$ when $\ops$ contains only the boolean semiring $\sem B$.

 We write $\phi(\tup x)$ to indicate
 that $\phi$ has free variables \emph{contained} in $\tup x$.
 The semantics of an $\sem S$-valued formula $\phi(\tup x)$  in a given structure $\str A$ is a function $\phi_{\str A}\from\str A^{\tup x}\to \sem S$,
 and is defined inductively, as expected.

  \paragraph{Counting logics.}
  The papers \cite{Kuske:2017:FLC:3329995.3330068,DBLP:conf/pods/GroheS18}
  study a logic denoted $\textrm{FOC}(\mathbb P)$,
  where $\mathbb P$ is a set of \emph{numerical predicates}
  of the form $\mathsf P\from \Z^k\to\sem B$.
  This logic can be seen as the  fragment of our logic
  $\textrm{FO}[\mathbb P\cup\set{\sem B,\Z}]$,
 where
      summation in $\Z$ is restricted to \emph{counting terms} of the form
      $\#_{\tup x}\phi$, in our syntax corresponding to $\sum_{\tup x}[\phi]_{\Z}$. 
As observed by \GS~\cite{DBLP:conf/pods/GroheS18}, the problem of evaluation
 of $\textrm{FOC}(\mathbb P)$
on the class of trees is as hard as the problem of evaluation of
 first-order logic on arbitrary graphs. In particular,
 under common complexity-theoretic assumptions,
 there is no algorithm which decides whether a given sentence $\phi\in\textrm{FOC}(\mathbb P)$
 holds in a given tree $T$, whose running time
 is $\Oof(|T|^c)$, where $c\in\N$ is independent of $\phi$.
Due to this, \GS propose to study a restricted fragment of the logic $\textrm{FOC}(\mathbb P)$,
denoted $\textrm{FOC}_1(\mathbb P)$.
In this fragment, numerical predicates in $\mathbb P$ can be applied
 only if they yield a formula with at most one free variable.
So for example, $\mathsf P(\phi(x),\psi(y))$ is no longer allowed as a formula, but $\mathsf P(\phi(x),\psi(x))$ is.
\GS proved that model-checking of $\textrm{FOC}_1(\mathbb P)$ has an almost-linear time algorithm on nowhere dense classes.

 \paragraph{The restricted fragment $\foops$.}
 We introduce an analogous fragment of $\fops$,
in which the following restriction is
 imposed on how connectives $c\in\ops$ can be used. The syntax is the same as for $\fops$,
with the difference that the construct $c(\phi^1,\ldots,\phi^k)$ is replaced by 
$[R(x_1,\ldots,x_l)]_{\sem S}\cdot c(\phi^1,\ldots,\phi^k)$
where in the \emph{guard} $[R(x_1,\ldots,x_l)]_\sem S$, the symbol $R$ is a $\sem B$-relation symbol in $\Sigma$ and $\set{x_1,\ldots,x_l}$ contains all the free variables of  $\phi^1,\ldots,\phi^k$.
 For example, all formulas discussed in the introduction,
 as well as  $\sum_x\sum_y [E(x,y)]_\Q\cdot \textit{weight}(x)/\textit{weight}(y)$,
 are in $\foops$, whereas the formula
 $\sum_x\sum_y\textit{weight}(x)/\textit{weight}(y)$ is not (here $/\from \Q\times\Q\to \Q$ is a connective).

Note that $\foops$
still extends first-order logic,
and also FAQ queries
over one semiring (without product aggregates).
Also, if $\mathbb P$ is a set of numerical predicates,
then
$\mathrm{FO_G}[\mathbb P\cup\set{\sem B,\Z}]$ strictly extends the logic  $\textrm{FOC}_1(\mathbb P)$ studied in~\cite{DBLP:conf/pods/GroheS18}, as it allows unrestricted
summation over $\Z$, rather than counting.

The main result of this paper is the following:
\begin{theorem}[\appmark]
  \label{thm:foc}
  Let $\ops$ be a collection of semirings and connectives, in which all operations take constant time,  let $\Sigma$ be a $\ops$-signature,
   let $\CCC$ be a class of $\Sigma$-structures of bounded expansion, and let $\phi(\tup x)\in \foops$ be a $\Sigma$-formula.

  There is an algorithm which, given a structure $\str A\in\CCC$,
  computes a data structure allowing to query
the value $\phi_\str A(\tup a)$,
   given a tuple $\tup a\in\str A^\tup x$.
  The algorithm runs in time $\Oof_\phi(|\str A|\log |\str A|)$ in general, and in time $\Oof_\phi(|\str A|)$ if all the semirings in $\ops$ are either rings or finite. The query time is $\Oof_\phi(\log|\str A|)$ in general, and $\Oof_\phi(1)$ if the output semiring of $\phi$ is either a ring or is finite.

  Moreover,  $\phi(\tup x)$ has output in  the boolean semiring,
  then the data structure also provides a constant-delay enumerator for~$\phi_\str A$.
\end{theorem}
\begin{proof}[\ifdoublecolumn{Proof }sketch]
To prove the theorem, we proceed by induction on the size of the formula $\phi(\tup x)\in\foops$. In the inductive step,
replace every guarded connective $[R(x_1,\ldots,x_l)]\cdot c(\phi^1,\ldots,\phi^k)$ occurring in $\phi$ at the top-most level (i.e. not within the scope of another connective) by a new weight symbol $r(x_1,\ldots,x_l)$. This yields a formula with no connectives, which can be viewed as a weighted expression $f$
in the sense of Section~\ref{sec:circuits}, or a boolean formula if $\phi$ is $\sem B$-valued.
Compute the new weights 
$r_\str A(x_1,\ldots,x_l)=[R_\str A(x_1,\ldots,x_l)]\cdot c(\phi_\str A^1,\ldots,\phi^k_\str A)$ in $\str A$ using the inductive assumption applied to $\phi^1,\ldots,\phi^k$,
and then evaluate $f$ with the newly computed weights, using Theorem~\ref{thm:algo-main}, or Theorem~\ref{thm:qe} if $\phi$ is $\sem B$-valued.
For the last part, if $\phi(\tup x)$ has output in the boolean semiring,
we apply Theorem~\ref{thm:enum} to the boolean formula $f(\tup x)$ considered above, yielding a constant-delay enumerator for~$\phi_\str A$.
\end{proof}

\section{Conclusion}

We presented an algebraic framework which allows to derive some existing and some new results 
concerning the evaluation, enumeration, and maintenance of query answers on sparse databases.
The main advantage of our framework is that it allows to achieve new, efficient algorithms, 
by simple algebraic considerations about the permanent in various semirings. 
Proving those results directly would be considerably more involved. In particular, Theorem~\ref{thm:foc} provides an efficient algorithm for enumerating answers to queries from a complex query language involving aggregates from multiple semirings. Our results partially resolve 
two questions posed in~\cite[Section 9]{DBLP:conf/pods/GroheS18}: we can answer aggregate queries with the usual SQL aggregates, and we can enumerate the answers to such queries, for classes of bounded expansion.

In this paper, we only study the basic applications of the main technical result, Theorem~\ref{thm:circuits}.
We believe that following these principles will allow to extend the  results in further directions,
perhaps by considering more intricate semirings and properties of their permanents. 
For example,  extending Theorem~\ref{thm:foc} to handle updates will require analyzing the interaction between the allowed connectives and the semiring operations. 
Furthermore, in our dynamic enumeration algorithm, Theorem~\ref{thm:enum}, we only handle updates which preserve the Gaifman graph. Arbitrary updates could be handled if we could have a dynamic algorithm for maintaining a low treedepth decomposition of the Gaifman graph. 
Such an algorithm would immediately allow to lift Theorem~\ref{thm:enum} to arbitrary updates which leave the structure in the considered class of databases $\CCC$. 
For some graph classes $\CCC$, such as classes of bounded degree or bounded treewidth, such an algorithm can be  trivially obtained. This allows, e.g., to derive the result of~\cite{DBLP:journals/tods/BerkholzKS18} (without mod quantifiers) from Theorem~\ref{thm:enum}. 

Another question which arises is about generalizing Theorem~\ref{thm:circuits}
to other graph classes and logics. 
In particular -- to nowhere dense graph classes, which are more general than bounded expansion classes,
and to classes of bounded treewidth, which are less general, but come with the more powerful MSO logic.
For nowhere dense classes, no analogue of the quantifier elimination result, Theorem~\ref{thm:qe} is known. Our main technical result, Theorem~\ref{thm:circuits}, is closely related to quantifier-elimination. Indeed, it is not difficult to derive Theorem~\ref{thm:qe} from the quantifier-free case of Theorem~\ref{thm:circuits} applied in the boolean semiring (where summation still allows introducing existential quantifiers).
Hence, generalizing Theorem~\ref{thm:circuits} to nowhere dense classes is closely connected to obtaining a quantifier elimination procedure there.



 
 

\bibliographystyle{abbrv}
\bibliography{ref}

\begin{thebibliography}{10}

\bibitem{DBLP:conf/icalp/AmarilliBJM17}
A.~Amarilli, P.~Bourhis, L.~Jachiet, and S.~Mengel.
\newblock A circuit-based approach to efficient enumeration.
\newblock In {\em {ICALP} 2017}, pages 111:1--111:15, 2017.

\bibitem{DBLP:conf/icdt/AmarilliBM18}
A.~Amarilli, P.~Bourhis, and S.~Mengel.
\newblock Enumeration on trees under relabelings.
\newblock In {\em {ICDT} 2018}, pages 5:1--5:18, 2018.

\bibitem{DBLP:journals/tods/BerkholzKS18}
C.~Berkholz, J.~Keppeler, and N.~Schweikardt.
\newblock Answering {FO+MOD} queries under updates on bounded degree databases.
\newblock {\em {ACM} Trans. Database Syst.}, 43(2):7:1--7:32, 2018.

\bibitem{Bjorklund:2010:EPR:1840004.1840168}
A.~Bj\"{o}rklund, T.~Husfeldt, P.~Kaski, and M.~Koivisto.
\newblock Evaluation of permanents in rings and semirings.
\newblock {\em Inf. Process. Lett.}, 110(20):867--870, Sept. 2010.

\bibitem{DBLP:conf/ijcai/Darwiche99}
A.~Darwiche.
\newblock Compiling knowledge into decomposable negation normal form.
\newblock In {\em {IJCAI} 99}, pages 284--289, 1999.

\bibitem{DBLP:conf/focs/DvorakKT10}
Z.~Dvor{\'{a}}k, D.~Kr{\'{a}}l, and R.~Thomas.
\newblock Deciding first-order properties for sparse graphs.
\newblock In {\em {FOCS 2010}}, pages 133--142, 2010.

\bibitem{Dvorak2013testing}
Z.~Dvo{\v{r}}{\'a}k, D.~Kr{\'a}l, and R.~Thomas.
\newblock Testing first-order properties for subclasses of sparse graphs.
\newblock {\em Journal of the ACM (JACM)}, 60(5):36, 2013.

\bibitem{lsd}
J.~Gajarsk{\'{y}}, S.~Kreutzer, J.~Ne\v{s}et\v{r}il, P.~O. de~Mendez,
  M.~Pilipczuk, S.~Siebertz, and S.~Toru{\'n}czyk.
\newblock First-order interpretations of bounded expansion classes.
\newblock In {\em {ICALP} 2018}, pages 126:1--126:14, 2018.

\bibitem{provenance}
T.~J. Green, G.~Karvounarakis, and V.~Tannen.
\newblock Provenance semirings.
\newblock In {\em PODS 2007}, pages 31--40, 2007.

\bibitem{grohe2011methods}
M.~Grohe and S.~Kreutzer.
\newblock Methods for algorithmic meta theorems.
\newblock {\em Model Theoretic Methods in Finite Combinatorics}, 558:181--206,
  2011.

\bibitem{gks}
M.~Grohe, S.~Kreutzer, and S.~Siebertz.
\newblock Deciding first-order properties of nowhere dense graphs.
\newblock In {\em STOC 2014}, pages 89--98. ACM, 2014.

\bibitem{DBLP:conf/pods/GroheS18}
M.~Grohe and N.~Schweikardt.
\newblock First-order query evaluation with cardinality conditions.
\newblock In {\em PODS 2018}, pages 253--266, 2018.

\bibitem{Har-PeledQ17}
S.~Har{-}Peled and K.~Quanrud.
\newblock Approximation algorithms for polynomial-expansion and low-density
  graphs.
\newblock {\em {SIAM} J. Comput.}, 46(6):1712--1744, 2017.

\bibitem{KazanaS13}
W.~Kazana and L.~Segoufin.
\newblock Enumeration of first-order queries on classes of structures with
  bounded expansion.
\newblock In {\em PODS 2013}, pages 297--308, 2013.

\bibitem{Kuske:2017:FLC:3329995.3330068}
D.~Kuske and N.~Schweikardt.
\newblock First-order logic with counting: At least, weak hanf normal forms
  always exist and can be computed!
\newblock In {\em LICS 2017}, pages 73:1--73:12, 2017.

\bibitem{nevsetvril2008grad}
J.~Ne{\v{s}}et{\v{r}}il and P.~Ossona~de Mendez.
\newblock Grad and classes with bounded expansion {I}. {D}ecompositions.
\newblock {\em European Journal of Combinatorics}, 29(3):760--776, 2008.

\bibitem{nevsetvril2008gradb}
J.~Ne{\v{s}}et{\v{r}}il and P.~Ossona~de Mendez.
\newblock Grad and classes with bounded expansion {II}. {A}lgorithmic aspects.
\newblock {\em European Journal of Combinatorics}, 29(3):777--791, 2008.

\bibitem{sparsity}
J.~Ne\v{s}et\v{r}il and P.~{Ossona de Mendez}.
\newblock {\em Sparsity --- {G}raphs, {S}tructures, and {A}lgorithms},
  volume~28 of {\em Algorithms and combinatorics}.
\newblock Springer, 2012.

\bibitem{Olteanu:2012:FRQ:2274576.2274607}
D.~Olteanu and J.~Z\'{a}vodn\'{y}.
\newblock Factorised representations of query results: Size bounds and
  readability.
\newblock In {\em ICDT 2012}, pages 285--298, 2012.

\bibitem{DBLP:conf/pods/SchweikardtSV18}
N.~Schweikardt, L.~Segoufin, and A.~Vigny.
\newblock Enumeration for {FO} queries over nowhere dense graphs.
\newblock In {\em PODS 2018}, pages 151--163, 2018.

\end{thebibliography}
\cleardoublepage

\appendix

\section{Proof Theorem~\ref{thm:circuits}}
\label{app:main-proof}
In this section we prove Theorem~\ref{thm:circuits}. 
To this end, we need to describe, for each class $\CCC$ of bounded expansion and closed quantifier-free $\Sigma(\tup w)$-expression $f$, a construction of a class of circuits with bounded expansion. This will be achieved by a composition of a series of steps, in each step considering more complicated classes $\CCC$.

\subsection{Proof outline}
To prove Theorem~\ref{thm:circuits}, we will 
show a series of reductions depicted in Figure~\ref{fig:red}, and prove that the statement of the theorem holds when $\CCC$ is a class of labelled forests of bounded depth. 
The notion of reduction used here is made formal below. 
The reductions depicted in the figure will be shown in subsequent sections.

A \emph{binary relational structure} is a  structure over a signature consisting of unary and binary relation symbols.

 \newcommand{\reduce}{$\rightarrow$}
\newcommand{\prop}[1]{\begin{tabular}{c}{#1}\end{tabular}}

\begin{figure*}[h!]
	\centering
	\ifsinglecolumn{\footnotesize}
	\begin{tabular}{ccccccc}
	\begin{tabular}{c}bounded\\expansion\end{tabular}&\reduce& 
	\begin{tabular}{c}bounded expansion\\binary relational\end{tabular}&$\Rightarrow$&
	\begin{tabular}{c}bounded treedepth\\binary relational\end{tabular}&\reduce&
	\begin{tabular}{c}labelled forests\\of bounded depth\end{tabular}
	\end{tabular}
	\caption{Outline of the proof of Theorem~\ref{thm:circuits}. 
	An arrow $X\rightarrow Y$ signifies that every class of structures  with property $X$ reduces to some class with property $Y$. $X\Rightarrow Y$ is a different kind of reduction, using the fact that classes of bounded expansion admit low-treedepth colorings.}
	\label{fig:red}
\end{figure*}

\paragraph{Reductions.}

We say that a class $\CCC$ \emph{reduces} to $\DDD$ if
for every atomic formula 
$\alpha(\tup x)$ of the form
$f(x_1,\ldots,x_k)=y$ or $R(x_1,\ldots,x_k)$,
where $x_1,\ldots,x_k,y$ are variables in $\tup x$ and $f$ and $R$ are function/relation symbols 
in the signature of $\CCC$, there is a quantifier-free formula $\alpha'(\tup x)$ in the signature of $\DDD$, and
there is a linear algorithm which inputs a  structure $\str A\in \CCC$ and computes a structure $\str A'\in \DDD$, such that
 $\alpha_{\str A}=\alpha'_{\str A'}$ for all $\alpha$ as above.

 \begin{lemma}\label{lem:reduction}
	If $\CCC$ reduces to $\DDD$ and the statement of Theorem~\ref{thm:circuits} holds for $\DDD$, then it holds for $\CCC$, where in both cases, we assume $f$ is quantifier-free and the arities of all weight functions in $\tup w$ to be~$1$.
\end{lemma}

Before proving Lemma~\ref{lem:reduction}, we observe  that it is enough to assume in the statement of Theorem~\ref{thm:circuits} that the expression $f$ is \emph{simple}, i.e., 
\begin{itemize}
	\item for every expression $w(t)$ occurring in $f$, where $w\in \tup w$ and $t$ is a term, $t$ is a single variable,
	\item for every expression $[\alpha]$ occurring in $f$, $\alpha$ is a literal of the form $R(x_1,\ldots,x_k)$ 
	or $\neg R(x_1,\ldots,x_k)$ 
	for some relation symbol 
	$R\in\Sigma\cup\set=$ 
	 or $f(x_1,x_2,\ldots,x_k)=y$ for some function symbol $f$, where $x_1,\ldots,x_k,y$ are variables.
\end{itemize}

\begin{lemma}\label{lem:simplify}
	Every  $\Sigma(\tup w)$-expression $f$
	is equivalent to a simple expression $f'$.
\end{lemma}
\begin{proof}
We apply a series of simplifications to $f$.

In the first step, we remove disjunctions and conjunctions from expressions $[\alpha]$ occurring in $f$. To this end, note that $[\alpha\land \beta]$ is equivalent to $[\alpha]\cdot [\beta]$,
whereas $[\alpha\lor \beta]$ is equivalent to 
$[\alpha][\beta]+[\alpha][\neg\beta]+[\neg\alpha][\beta]$. 
Hence, we may assume that every expression occurring in $\alpha$ is a literal, i.e. an atom $R(t_1,\ldots,t_k)$ or its negation, where $R\in \Sigma\cup\set{=}$, and $t_1,\ldots,t_k$ are terms. 

Next, observe that the expression $[R(t_1,\ldots,t_k)]$ is equivalent to the expression 
$$\sum_{x_1,\ldots,x_k}[t_1=x_1]\cdots [t_k=x_k][R(x_1,\ldots,x_k)].$$ The same holds if $R$ is replaced by $\neg R$. 
Further, each expression $[t=x]$ as above can be iteratively simplified, if $t$ has nesting depth larger than $1$, as follows:
if $t=f(t_1,\ldots,t_k)$ then $[t=x]$ is equivalent to $\sum_{y_1,\ldots,y_k}[y_1=t_1]\cdots [y_k=t_k][f(y)=x]$. Note that $t_1,\ldots,t_k$ have smaller depth than $t$.

After performing all these simplifications, we arrive at an expression $f$ where all boolean formulas $\alpha$ occurring in $f$ are of the form $R(x_1,\ldots,x_k)$, $\neg R(x_1,\ldots,x_k)$, or  $f(x_1,\ldots,x_k)=y$.

Finally, each expression $w(t)$, where $w\in \tup w$ and $t$ is a term, can be replaced by $\sum_x [t=x]\cdot w(x)$,
and then $[t=x]$ can be simplified as described above.
\end{proof}

\begin{proof}[\ifdoublecolumn{Proof }of Lemma~\ref{lem:reduction}]
	Assume that the statement of Theorem~\ref{thm:circuits} (for $f$ quantifier-free) holds for $\DDD$ and that $\CCC$ reduces to $\DDD$.
 We prove that the statement of Theorem~\ref{thm:circuits} holds for $\CCC$ and $f$ quantifier-free.

 Let $\Sigma$ be the signature of $\CCC$ and let $f$ be a closed $\Sigma(\tup w)$-expression. By Lemma~\ref{lem:simplify} we may assume that $f$ is simple. 
 Define a $\Sigma'(\tup w)$-expression $f'$ by replacing in $f$ each atomic boolean formula $\alpha$
 by the equivalent $\Sigma'$-formula $ \alpha'$,
 where $\alpha'$ is as in the definition of a reduction. Then $f'$ is quantifier-free and  $f_{\str A}=f'_{\str A'}$, where $\str A'$ is the output of the algorithm as in the definition of a reduction.
 
 Let $(D_{\str A'})_{\str A'\in \DDD}$ be a linear-time computable family of circuits evaluating $f'$ for the class $\DDD$.
 Given a structure $\str A\in\CCC$, define the circuit $C_\str A$ as $D_\str A'$,
 where $\str A'$ is computed by the algorithm from the reduction. 
Then $(C_\str A)_{\str A\in \CCC}$ is a linear-time computable family of circuits evaluating $f$ for the class $\CCC$.
\end{proof}





\subsection{Forests of bounded depth}
In this section, we prove a special case of Theorem~\ref{thm:circuits}, in the case when $\CCC$ is a  class of labelled, rooted forests of bounded depth. We treat such a forest as a structure over 
a signature consisting of unary relation symbols for representing the labels,  a unary relation symbol $\rt$ which is interpreted in a rooted forest  $F$ as the set of roots, and one unary function symbol $\parent$, which is interpreted in a forest $F$ as the function mapping every non-root node to its parent, and the roots to themselves.

\begin{lemma}\label{lem:circuits-trees}
	The statement of Theorem~\ref{thm:circuits} holds when $\CCC$ is a class of labeled forests of bounded depth, and $f$ is quantifier-free.
\end{lemma}

In the remainder of this section, 
fix a number $d\in \N$ and signature $\Sigma$ as above. By \emph{forest} we mean
a $\Sigma$-structures which is a forest, as described above. Fix a $\Sigma(\tup w)$-expression $f$.



\paragraph*{Basic expressions.}
A \emph{shape} with free variables $\tup x$
is a forest $F$ and
a tuple  $\tup a \in F^{\tup x}$.
We write $F(\tup x)$ for the shape defined by $F$ and $\tup a$, and
 write $F[x]$ instead of $\tup a[x]$ for $x\in \tup x$,
and $F[\tup x]$ for $\setof{F[x]}{x\in \tup x}$.
We require that every node in a shape $F$ is an ancestor of some node in $F[\tup x]$. 
We distinguish shapes only up to isomorphism,
i.e., we consider two shapes $F,G$ equal
if there is a forest isomorphism from $F$ to $G$
which maps $F[x]$ to $G[x]$,
for each $x\in \tup x$.

The \emph{atomic type} of a shape $F$
	is the set of all literals 
	which hold in $F$.
	Specifically, those literals are:
	\begin{align}\label{eq:label}
	  U(\parent^i(x))
	\end{align} 
	   for $0\le i\le d$, $x\in \tup x$, and unary predicate $U\in\Sigma$, such that $U(\parent^i(F[x]))$ holds,
	   \begin{align}
		\label{eq:eq}
	\parent^i(x)=\parent^j(y),
	   \end{align}  for $0\le i,j\le d$ and $x,y\in \tup x$ such that $\parent^i(F[x])=\parent^j(F[y])$	   
	 \begin{align}  
	\parent^i(x)\neq\parent^j(y), \label{eq:neq}
	\end{align}
	 for $0\le i,j\le d$ and $x,y\in \tup x$ such that $\parent^i(F[x])\neq\parent^j(F[y])$.

	With each shape $F(\tup x)$ we associate the $\Sigma$-formula, also denoted $F(\tup x)$, defined as the conjunction of all literals in the atomic type of $F$.
The following is immediate.
\begin{lemma}\label{lem:inj}
	Given a forest $\str A$  and a tuple $\tup a\in \str A^{\tup x}$, the formula
	$F_{\str A}(\tup a)$ holds iff there is
an embedding $\alpha\from F\to \str A$
of rooted forests which maps $F[x]\in F$ to $\tup a[x]\in \str A$, for $x\in\tup x$. If such an embedding exists, it is unique.
\end{lemma}

A \emph{labelled shape} is a pair $(F(\tup x),\lambda)$, where $F(\tup x)$ is a shape and
$\lambda$ associates to each node $v$ of $F$ 
a multiset $\lambda(v)$ of weight  symbols from $\tup w$.
With each labelled shape $(F(\tup x),\lambda)$ we associate a \emph{basic expression}, defined as follows.
For each node $v$ of $F$ arbitrarily choose a leaf $x$  of $F(\tup x)$ and number $i\ge 0$ such that $v=\parent^i(x)$. Let $\lambda_v(\tup x)$ be the expression $\prod_{w\in\lambda(v)}w(\parent^i(x))$.
Define the $\Sigma(\tup w)$-expression $F^\lambda(\tup x)$ 
as the product of the expression $[F(\tup x)]$
	 and of the all  expressions $\lambda_v(\tup x)$, for $v$ ranging over the nodes of $F$. We call $F^\lambda(\tup x)$ the \emph{basic expression} associated to $(F(\tup x),\lambda)$.

The following lemma is an immediate consequence of Lemma~\ref{lem:inj}.
\begin{lemma}\label{lem:basic}
	Let $(F(\tup x),\lambda)$ be a 
	labelled shape and $f=\sum_{\tup x} F^\lambda(\tup x)$. 
	For every  forest $\str A$ and  tuple of weight functions $\tup w$,
	$$f_{\str A(\tup w)}=\sum_{\alpha\from F\to \str A}
	\prod_{v\in F}\lambda_v(\alpha(v)),$$
	where the sum ranges over all embeddings of rooted forests.
\end{lemma}

An \emph{$\sem S$-combination} of expressions $\alpha_1,\ldots,\alpha_k$ is an expression of the form $s_1\cdot \alpha_1+\ldots+s_k\cdot \alpha_k$, for $s_1,\ldots,s_k\in \sem S$.
We say that two expressions $f,g$ are \emph{equivalent} over a class $\CCC$ if $f_{\str A(\tup w)}=g_{\str A(\tup w)}$ for every $\str A\in\CCC$. We also say in this case that the \emph{identity $f=g$ holds} (over $\CCC$).

\begin{lemma}\label{lem:basic-decomposition}
  Every sum-free $\Sigma(\tup w)$-expression $\alpha(\tup x)$  is equivalent over $\CCC$ to an $\sem S$-combination of basic expressions.
\end{lemma}
\begin{proof}
		Fix $\tup x$. 
 The following identities  hold over $\CCC$:
	  \begin{align}
	  [\parent^i(x)=\parent^j(y)]+[\parent^i(x)\neq \parent^j(y)]&=1\label{eq:parents}
	  \end{align}
	  \begin{align}
[\rt(x)]+[\rt(\parent(x))]+[\rt(\parent^2(x))]+\ifdoublecolumn{\\}\cdots+[\rt(\parent^d(x))]&=1,\label{eq:depths}
\end{align}
\begin{align}
[\tau_1(\parent^i(x))]+[\tau_2(\parent^i(x))]+\ldots+[\tau_k(\parent^i(x))]=1\label{eq:colors}
	  \end{align}
	  where $0\le i,j\le d$ and $x,y\in \tup x$,
	  and $\tau_1(x),\ldots,\tau_k(x)$ 
	  are all formulas of the form 
	  $$\bigwedge_{U\in \Sigma_+} U(x)\land 
	  \bigwedge_{U\in \Sigma_-} \neg U(x),$$
	  where $\Sigma_+$ and $\Sigma_-$ is a partition of the unary predicates in $\Sigma$.

	  Let $\gamma(\tup x)$ be the product of all the left-hand sides of the above 
	  identities. Hence, $\gamma(\tup x)=1$ is an identity that holds over $\CCC$,
	  and so is $\alpha(\tup x)\cdot \gamma(\tup x)=\alpha(\tup x)$.
	  Using distributivity, rewrite $\alpha(\tup x)\cdot \gamma(\tup x)$
	  as a sum of products. Then, each summand $P(\tup x)$
	 is a product of weight symbols applied to terms, and of expressions of the form $[\alpha]$, such that for each of the equations~\eqref{eq:parents},\eqref{eq:depths},\eqref{eq:colors}, exactly one of the summands on the left-hand side is present in the product $P(\tup x)$. It is not difficult to see that either $P(\tup x)$ 
	  is unsatisfiable in $\CCC$, in the sense that $P_{\str A(\tup w)}=0$ for every $\str A\in\CCC$, or is
	  of the form $F^\lambda(\tup x)$, for some labelled shape $(F(\tup x),\lambda)$.
	  This yields the statement.
\end{proof}

\begin{proof}[\ifdoublecolumn{Proof }of Lemma~\ref{lem:circuits-trees}]
	Without loss of generality we may assume that $f$ is a sum of expressions of the form $\sum_{\tup x}\psi(\tup x)$, where $\psi(\tup x)$ is sum-free.
By Lemma~\ref{lem:basic-decomposition} we may further assume that $\psi(\tup x)$ as an $\sem S$-combination of basic expressions. It is therefore enough to prove the lemma under the assumption that $\psi(\tup x)$ is a basic expression, since we can easily construct a circuit for an $\sem S$-combination of  expressions, given circuits for each of them.

Hence, it suffices to prove the lemma in the case when $f$ is of the form $\sum_{\tup x}F^\lambda(\tup x)$, for some labelled shape $(F(\tup x),\lambda)$  and variables $\tup x=\set{x_1,\ldots, x_k}$. Without loss of generality we may assume that in the shape $F(\tup x)$, every node $x\in\tup x$ is distinct; otherwise we can decrease $|\tup x|$.

The proof proceeds by induction on $d$. 
In the base case, assume that $d=0$.
If $F$ has depth larger than $0$ then $f=0$ over $\CCC$.
since there is no embedding $F\to \str A$.
Otherwise,  $F(\tup x)$ consists only of the roots $x_1,\ldots,x_k$. As for $\str A\in\CCC$, embeddings $\alpha\from F\to \str A$ correspond 
to sequences $a_1,\ldots,a_k$ of distinct elements of $\str A$,  by Lemma~\ref{lem:basic} the following identity holds over $\CCC$:
\begin{align}\label{eq:perm1}
	f=\sum_{\substack{v_1,\ldots,v_k\\\textit{distinct}}}\lambda_{x_1}(v_1)\cdots \lambda_{x_k}(v_k).	
\end{align}
This is the permanent of the matrix with $k$ rows and columns indexed by the vertices of $\str A$, where 
the entry $(i, v)$ has value $\lambda_{x_i}(v)=\prod_{w\in \lambda(x_i)}f(v)$.
This  naturally yields a circuit $C_\str A$, with one permanent gate connected whose inputs are outputs of bounded size circuits computing the products $\lambda_{x_i}(v)$, for each $v\in V$ and $1\le i\le k$.
Clearly, $C_\str A$ can be computed in linear time, given $\str A\in \CCC$.
	
In the inductive step, assume that $\CCC$ is the class of forests of depth at most $d+1$.
For a forest $\str A\in \CCC$ and its root $v$, by $\str A^v$ we denote the subforest of $\str A$ induced by the strict descendants of $v$ in $\str A$. 
Similarly, for a root $r\in  F$ of a shape $F(\tup x)$,
let $F^r$ denote the shape formed by the forest induced by the strict descendants of $r$ in $F$, and let
$\tup x_r$ denote the set of nodes in $F[\tup x]$ which belong to $F^r$, and let $\lambda^r$ denote the labelling $\lambda$ restricted to $F^r$.
Let $$f^r=\sum_{\tup x_r}(F^r)^{\lambda^r}(\tup x_r).$$

\begin{claim}\label{cl:permaforest}
The following holds for all $\str A\in\CCC$:
$$f=\sum_{\beta}\prod_{r}
 \lambda_r(\beta(r))\cdot f^r_{\str 
 A^{\beta(r)}}
$$
where $\beta$ ranges over all injective mappings of the roots of $F$ to the roots of $\str A$ and $r$ ranges over all roots of $F$.
\end{claim}
\begin{proof}
	We use Lemma~\ref{lem:basic} to express the left-hand side above. Every embedding $\alpha \from F\to \str A$ 
	can be constructed in two steps:
	first choose an injection $\beta$ of the roots of $F$ into the roots of $\str A$, and then, for each root $r$ of $F$, choose an embedding of $F^r$ into $\str A^{\alpha(r)}$.
	Hence, by Lemma~\ref{lem:basic}
	we get:
	\begin{align*}		
		f_{\str A(\tup w)}&=
		\sum_{\alpha\from F\to \str A}
		\prod_{v\in F}\lambda_v(\alpha(v))\\
		&=\sum_{\beta}\prod_r\left(\sum_{\alpha_r\from F^r\to \str A^{\beta(r)}}\prod_{v\in F}\lambda_v(\alpha(v))\right)\\&=
		\sum_{\beta}\prod_r\lambda_r(\beta(r))\left(\sum_{\alpha_r\from F^r\to \str A^{\beta(r)}}\prod_{v\in F^r}\lambda_v(\alpha(v))\right)=\\
		&=\sum_{\beta}\prod_r\lambda_r(\beta(r))\cdot f^r_{\str A^{\beta(r)}}.
	\end{align*}
		This yields the claim.
\end{proof}

To obtain the desired circuit $C_\str A$ computing $f$ on $\str A$,
 we apply the inductive assumption to each  expression $f^r$, yielding a linear time algorithm for each root $r$ of $F$,
 which
for each forest $H$ of depth $d$  
 produces a circuit $C^r_H$ computing $f^r$. 
Given a forest $\str A$ of depth $d+1$, we construct the circuit $C_\str A$ which (as a term) is  a permanent with rows corresponding to the roots $r$ of $F$ and columns corresponding to the roots $v$ of $\str A$, and where the entry at row $r$ and column $v$ is 
the product of $\lambda_r(v)$ and 
	  $C_{\str A^v}^r$.
It is easy to see that that this can be implemented in linear time.
Correctness of the construction follows from Claim~\ref{cl:permaforest}. 
\end{proof}

\subsection{Binary relational structures of bounded treedepth}
In this section, we prove a special case of Theorem~\ref{thm:circuits}, in the case when $\CCC$ is a class of structures over a binary relational signature.

\begin{lemma}Let $\CCC$ be a class binary relational structures of bounded treedepth. Then $\CCC$ reduces to a class of labelled rooted forests of bounded depth.
\end{lemma}

\begin{proof}
	Given a structure $\str A\in \CCC$, let $F_\str A$ be an arbitrarily chosen rooted spanning forest of the Gaifman graph of $\str A$. As  those Gaifman graphs have bounded treedepth, they also have  bounded path length,
 so $\setof{F_\str A}{\str A\in \CCC}$ is a class of forests of  depth bounded by some constant $d\in\N$. 

 Moreover, label every node $v$ of $F_\str A$ by the following predicates, where $\Sigma$ denotes the signature of $\CCC$:
\begin{itemize}
	\item a unary predicate $U_i$, where $0\le i\le d$ is the depth of $v$ in $F_\str A$,
	\item each unary predicate $U\in \Sigma$  such that $U_{\str A}(v)$ holds,
	\item for each binary symbol $R$ and $0\le i\le d$, label $a$ by a new unary predicate $R_i$ 	
	if $R_{\str A}(a,b)$ holds, where $b$ is the ancestor of $a$ at depth $i$ in $F_\str A$.
\end{itemize}
Finally, let $F_\str A$ be equipped with the  function mapping each non-root node to its parent, and the roots to themselves.
This turns $F_\str A$ into a labelled forest. 
Let $\FFF=\setof{F_\str A)}{\str A\in \CCC}$ be the class of all such forests. 

Note that the following equivalences hold:
\begin{align*}
	R(a,b)&\iff 
	\bigvee_{0\le i\le j\le d} U_i(a)\land U_j(b)\land (a=\parent^{j-i}(b))\land R_i(b)\\
	&\lor\bigvee_{0\le i\le j\le d} U_i(b)\land U_j(a)\land (b=\parent^{j-i}(a))\land R_i(a).
\end{align*}
Hence, the function mapping  $\str A$ to $F_\str A$ yields a reduction from $\CCC$ to $\FFF$.
\end{proof}

By Lemma~\ref{lem:reduction}, we get:
\begin{corollary}\label{cor:circuits-treedepth}
    The statement of Theorem~\ref{thm:circuits} holds when $\CCC$ is a class of binary relational structures of bounded treedepth and $f$ is quantifier-free
    and all weight functions have arity $1$.
\end{corollary}

\subsection{Classes of binary relational structures of bounded expansion}
\begin{lemma}\label{lem:circuits-reduction}
	The statement of Theorem~\ref{thm:circuits} holds when $\CCC$ is a class of binary relational structures of bounded expansion and $f$ is quantifier-free.
\end{lemma}
\begin{proof}Let $f$ be a closed quantifier-free $\Sigma(\tup w)$-expression. We show that there is a linear time algorithm producing circuits that evaluate $f$ over $\CCC$.
It is enough to consider the case when $f$ is of the form $\sum_{\tup x}g(\tup x)$ for some sum-free $\Sigma(\tup w)$-expression $g$ with free variables $\tup x$. Let $p=|\tup x|$.

As $\CCC$ has bounded expansion, there 
is a number $d\in\N$ and a finite set of colors $C$ such that every structure $\str A\in \CCC$ has a coloring $\gamma\from \str A\to C$ with colors from $C$ such that 
for any set $D\subset C$ of at most $p$ colors, the subgraph $G_D$ of the Gaifman graph  of $\str A$ induced by the set $\gamma^{-1}(D)$ of vertices of color in $D$
has treedepth bounded by $d$.
Moreover, we can assume that the coloring $c$ as above can be computed in linear time, 
given $\str A\in \CCC$.

Given a $\Sigma$-structure $\str A$,
first define its extension $\str A'$ by adding unary predicates $U_c$, for each $c\in C$, where $U_c$ marks the vertices of color $c$, i.e. $\gamma^{-1}(\set c)$. 
For an $\tup x$-tuple of colors $\tup c\in C^\tup x$,
define the formula testing that each variable  $x$ gets color $\tup c(x)$ under $f$:
$$\tau_\tup c(\tup x)=\bigwedge_{x\in \tup x}U_{\tup c[x]}(x).$$
Note that for every tuple $\tup a\in \str A^{\tup x}$ there is exactly one $\tup c\in C^\tup x$ such that $\tau_\tup c(\tup a)$ holds
(namely, $\tup c$ satisfying $\tup c[x]=\gamma(\tup a[x]))$ for $x\in\tup x$).
Hence, we have the following identities:
\begin{align}
\sum_{\tup x}\psi(\tup x)&=
\sum_{\tup c \in C^\tup x}\sum_{\tup x}
[\tau_\tup c(\tup x)]\cdot \psi(\tup x)\ifdoublecolumn{\\}
=\ifdoublecolumn{&}\sum_{\substack{D\subset C\\ |D|\le p}}\quad
\underbrace{\sum_{\tup x}\ \sum_{\substack{\tup c\in D^\tup x\\\textit{surjective}}}\quad  [\tau_\tup c(\tup x)]\cdot\psi(\tup x)}_{f^D}.
\end{align}
For a set $D\subset C$ of at most $p$ colors,  let $\str A^D$ denote the substructure of $\str A'$ induced by $G_D$, and let $f^D$ be the expression marked above.
From the above we have the following identity over $\CCC$: 
\begin{align}\label{eq:correct}
	f=\sum_{\substack{D\subset C\\ |D|\le p}}f^D.	
\end{align}

Let $\Sigma'$ be the signature extending $\Sigma$ by the unary predicates $U_c$, for $c\in C$.
Let $\DDD$ be the class of $\Sigma'$-structures
whose Gaifman graph has treedepth bounded by $d$.
Applying Corollary~\ref{cor:circuits-treedepth}
to $\DDD$ and $f^D$, for each $D\subset C$ of size at most $p$, we get a linear time algorithm computing circuits that evaluate $f^D$ on $\DDD$.

Given $\str A\in\CCC$, we construct a circuit $C_\str A$ computing $f_\str A$
as follows.
First compute the coloring $c\from \str A\to C$, producing a $\Sigma'$-structure $\str A'$ extending $\str A$ by the unary predicates $U_c$ marking the colors. Next, for each 
$D\subset C$ of size at most $p$, apply the linear time algorithm to the substructure of $\str A'$ induced by $c^{-1}(D)$, yielding a circuit $C^D_\str A$ computing $f^D$. The resulting circuit $C_\str A$ is the sum of the  the circuits $C^D_\str A$.  
Correctness of the construction follows from~\eqref{eq:correct}.

\end{proof}

\begin{corollary}\label{cor:binary}
	Theorem~\ref{thm:circuits} holds when $\CCC$ is a class of binary relational structures of bounded expansion and all weight functions have arity $1$.
\end{corollary}

\subsection{General case}
We now treat the general case of a class of structures $\CCC$ of bounded expansion, by reducing it to the case treated previously.

We first recall the notion of degeneracy.
A graph $G$ is \emph{$d$-degenerate} if its edges can be oriented yielding an acyclic orientation of out-degree bounded by $d$.
A class of graphs $\GGG$ has \emph{bounded degeneracy} if there is some $d\in\N$ such that every $G\in\GGG$ is $d$-degenerate. 
It is well-known that a class of graphs $\GGG$ has bounded degeneracy if and only if there is a bound $c$ such that for every  subgraph $H$ of a graph  $G\in \GGG$,
$H$ has at most $c\cdot |H|$ edges (cf. eg.~\cite{Dvorak2013testing}).
In particular, if $\GGG$ has bounded expansion, then $\GGG$ has bounded degeneracy.

The following lemma allows us to reduce statements concerning arbitrary signatures to analogous statements concerning  signatures with unary relation and function symbols only. 

\begin{lemma}
    \label{lem:unify}
	Let $\CCC$ be a class of $\Sigma$-structures of bounded expansion, where $\Sigma$ is a relational signature. Let $\tup w$ be a set of weight symbols.
    There is a signature $\Sigma'$ consisting of unary relation and function symbols,
    a set $\tup w'$ of unary weight symbols, a collection 
	of positive quantifier-free $\Sigma'$-formulas $(\alpha^R)_{R\in \Sigma}$ and $\Sigma'(\tup w')$-weighted expressions $(f^w)_{w\in \tup w}$,
	and an algorithm which inputs a $\Sigma(\tup w)$-structure $\str A(\tup w)$ with $\str A\in \CCC$ 
    and computes in linear time a $\Sigma'(\tup w')$-structure $\str A'$ with the same Gaifman graph as $\str A$, such that 
\begin{align*}
    R_{\str A}&=\alpha^R_{\str A'}&\text{ for each $R\in \Sigma$},\\
    w_\str A&=f^w_{\str A'(\tup w')}&\text {for each $w\in \tup w$}.
\end{align*}
\end{lemma}

\begin{proof}[\ifdoublecolumn{Proof }of Lemma~\ref{lem:unify}]
Since $\CCC$ has bounded expansion, the class of its Gaifman graphs is $d$-degenerate, for some $d\in\N$. 
The signature $\Sigma'$ contains unary function symbols $f_1,\ldots,f_d$, as well as unary predicates $R_\tup t$, for every relation symbol $R\in\Sigma$ of arity $k$ and tuple $\tup t\in [d]^k$.

Given a structure $\str A\in \CCC$, fix an acyclic orientation  of the edges of  the Gaifman graph of $\str A$ with out-degree at most $d$. It is well-known that such an orientation can be computed in linear time, by a greedy algorithm.

For $a\in\str A$ and $1\le i\le d+1$,  let $f_i(a)$ be the $i$th out-neighbor 
of $a$, if it exists, and $a$ otherwise.
For $R\in \Sigma$ of arity $k$ and  $\tup t\in [d+1]^k$,
where  $[d+1]=\set{1,\ldots,d+1}$,
let $R_\tup t(a)$ hold if $R(f_{\tup t(1)}(a),\ldots,f_{\tup t(k)}(a))$ holds.

Define $\str A'$ as the resulting $\Sigma'$-structure consisting of the domain of $\str A$, the unary functions $f_1,\ldots,f_d$ and unary relations $R_\tup t$ defined above. Then $\str A'$ has the same Gaifman graph as $\str A$.

Note that for every relation $R\in \Sigma$ and tuple $\tup a\in R_\str A$, all vertices in $\tup a$ are contained in a clique, and hence, as the chosen orientation is acyclic, there is a (unique) vertex
$a$ in $\tup a$ such that  all the  vertices in $\tup a$ are of  the form $f_i(a)$, for some (unique) $1\le i\le {d+1}$. 
Hence, $$R(a_1,\ldots,a_k)\iff \bigvee_{1\le i\le k}\bigvee_{\tup t\in [d+1]^k} R_\tup t(a_i)\land \bigwedge_{1\le j\le k} a_j=f_{\tup t(j)}(a_i).$$

Similarly, for $w\in \tup w$ of arity $k$ and $\tup t\in [d+1]^k$,
define a weight function 
by $$w_\tup t(a)=w(f_{\tup t(1)}(a),\ldots,f_{\tup t(k)}(a)).$$
Then the following equality holds: $$w(a_1,\ldots,a_k)=\sum_{1\le i\le k}
\sum_{\tup t\in [d+1]^k} w_\tup t(a_i)\cdot  \prod_{1\le j\le k} [a_j=f_{\tup t(j)}(a_i)].$$

This finishes the proof of Lemma~\ref{lem:unify}.
\end{proof}





Theorem~\ref{thm:circuits} now follows:

\begin{proof}[\ifdoublecolumn{Proof }of Theorem~\ref{thm:circuits}]
    We  assume that $\Sigma$ contains no function symbols, as those can be encoded by relation symbols in the usual way. 
    Hence, $\CCC$ is a class of relational structures which has bounded expansion. Applying Lemma~\ref{lem:unify} we may in turn assume that $\Sigma$ consists of unary  relation and function symbols, and that all the weight symbols in $\tup w$ have arity $1$.
    Using the same approach as above, this can be further reduced to a class of binary relational structures, by replacing functions by their graphs. The theorem follows by Corollary~\ref{cor:binary}.
\end{proof}

\section{Remaining proofs}
In this Appendix, we provide the missing details for Sections 4-7.

\subsection{Proof of Lemma~\ref{lem:perm-rings}}
\begin{proof}
    For a fixed number of rows $k$, we use the following inclusion-exclusion formula:
    
    \begin{multline*}
        \sum_{\substack{x_1,\ldots,x_k\\\textit{distinct}}}w_1(x_1)\cdots w_k(x_k)\ifdoublecolumn{=\\}
        =\sum_{x_1,\ldots,x_k}w_1(x_1)\cdots w_k(x_k)\prod_{i\neq j}(1-[x_i= x_j]).
    \end{multline*}

    Using distributivity, and then eliminating each variable $y$ occurring in an expression $[x=y]$, 
    allows to express the permanent as 
    a term $t$  depending on $k$ only, applied to expressions with one variable,
    of the form $$\sum_x w_{i_1}(x)\cdots w_{i_\ell}(x),$$ where $1\le i_1<\ldots<i_\ell\le k$.
    When $x$ in the sum ranges over a fixed set of $n$ elements, the above term can be viewed as a circuit $C_n$ which has  the properties required by the lemma.
\end{proof}

\subsection{Proof of Lemma~\ref{lem:counting}}
Fix a finite semiring $\sem S$.
By a \emph{boolean} we mean either of the elements $0,1\in \sem S$.
A \emph{test gate} is a gate parameterized by an element $s\in\sem S$ which inputs an element
$t$ of $\sem S$ and outputs the boolean $1\in \sem S$ if $s=t$ and $0\in\sem S$ otherwise. A \emph{mod} gate with \emph{modulus} $m\in\N$, inputs boolean values and outputs $1$ if the number of $1$'s on input is a multiple of $m$, and outputs $0$ otherwise. A \emph{threshold gate} has an associated threshold $t\in\N$, and given $n$ booleans on input, outputs $1$ if at least $t$ of them are $1$'s. A circuit with \emph{counting gates} is a  circuit without permanent gates, but with test gates, mod gates and threshold gates, and well as addition and product gates.

For an integer $n$ and element $s\in \sem S$,
let $n\cdot s$ denote the $n$-fold sum $s+\cdots +s$ if $n\ge 1$, and $0$ otherwise. 
We will use the following lemma, providing circuits computing $(\sum_i a_i-k)\cdot s$, where $a_1,\ldots,a_n$ are input booleans and $k\ge 0$ and $s\in\sem S$ are fixed.
Let $m_{\sem S}$ be the least common multiple of the orders of the cyclic groups contained in $(\sem S,+)$. For example, for the boolean semiring, $m_{\sem B}$ is $1$ and for the ring $\sem Z_k$ of integers modulo $k$, $m_{\sem Z_k}$ is $k$.

\begin{lemma}
    \label{lem:mod}
    For each fixed $s\in\sem S$ and  $k\in\N$ there is a family $(C_n)_{n\in \N}$ of circuits with counting gates,      
    where $C_n$ has $n$ boolean inputs and outputs 
    $(\ell-k)\cdot s$, where $\ell$ is the number of $1$'s among the inputs.
    The circuits $(C_n)_{n\in \N}$ have bounded depth, bounded fan-out, bounded expansion, each threshold is at most $\Oof(|\sem S|+k)$, and the modulus of each mod gate is the order of some cyclic subgroup of $(\sem S,+)$.
\end{lemma}

\begin{proof}[\ifdoublecolumn{Proof }sketch]
The idea is that by finiteness of $\sem S$, the sequence $(m\cdot s)_{m\ge 0}$ is
         ultimately periodic, where the  period forms a cyclic subgroup of $(\sem S,+)$.

More precisely, define a directed graph $G=(V,E)$ with $V=\set{0,s,2\cdot s,\ldots}\subset \sem S$
and  $E=\setof{(n\cdot s,(n+1)\cdot s)}{n\ge 0}$.
Since every vertex is reachable from $0\in V$ and every vertex has one out-neighbor, the graph $G$  is a ``lasso'' of the form: a directed path $P$ of length at most $|\sem S|-1$, a directed cycle $C$ disjoint from $P$, and an edge connecting the endpoint of $P$ with a vertex in $C$.

\begin{claim}
The cycle $C$ forms a cyclic subgroup of $\sem S$.
\end{claim}

\begin{proof}
    The set $\set{s,2\cdot s,3\cdot s,\ldots}$
    forms a finite sub-semigroup of $\sem S$. It is well-known that 
    every finite semigroup contains an 
    idempotent $e$, i.e., an element such that $e=e+e$ (specifically, $|\sem S|!\cdot s$ is such an idempotent). 
    Let $\omega>0$ be such that $e=\omega\cdot s$ is idempotent.
    Then $e$ belongs to the cycle $C$, since  $e=\omega\cdot s=(2\omega)\cdot s=(3\omega)\cdot s=\ldots$

    Then $C$ forms a subgroup of $(\sem S,+)$, with neutral element $e$ and inverse of $(\omega+\ell)\cdot s$ being 
    $({c\cdot \omega-\ell})\cdot s$, where $c\in\N $ is such that $\ell\le (c-1)\cdot\omega$.
    Moreover, $C$ is a cyclic group, generated by $(\omega+1)\cdot s$. 
\end{proof}
    
Let $n_0=|P|$. 
Then the elements $\setof{n\cdot s}{0\le n<n_0}$ are all the distinct elements of the path $P$,
and the elements $\setof{n\cdot s}{n_0\le n<n_0+|C|}$ are all the distinct elements of the cycle $C$.

Fix an element $t\in\sem S$. Let
$X_{t}=\setof{\ell\ge 0}{\ell\cdot s=t}$. It is easy to see that:
\begin{itemize}
    \item if $t\notin P\cup C$ then $X_t=\emptyset$,
    \item if $t\in P$ then 
    $X_t=\set n$, where 
     $0\le n<n_0$ is such that $n\cdot s=t$,
     \item if $t\in C$ then 
     $\setof{\ell\ge n_0}{(\ell-n_0)\mod |C|=p}$, where 
     $0\le p<|C|$ is such that $t=(\ell+p)\cdot s$.
\end{itemize}
For fixed $k\in\N$, denote by 
$X_t+k\subset \N$ the set
$\setof{\ell+k}{\ell\in X_t}$.
Then $X_t+k=\setof{\ell}{(\ell-k)\cdot s=t}$. For each of the three cases considered in the items above,
we can easily construct a family of $(C_n^t)_{n\ge 0}$ of circuits with threshold gates and mod gates, where $C_n^t$ inputs $n$ bits and outputs $1$ if  the number of $1$'s among its inputs belongs to  $X_t+k$. These circuits  use mod gates with modulus $|C|$ and threshold gates with threshold  $\Oof(|\sem S|+k)$.

    Having constructed such circuit families $C^t_n$ for each $t\in \sem S$, we construct a circuit $C_n$ whose output gate is the sum of the products $t\cdot C^t_n$, for $t\in \sem S$. The circuit family $C_n$ has the desired properties.
\end{proof}
\begin{proof}[\ifdoublecolumn{Proof }of Lemma~\ref{lem:counting}]
    Observe that the permanent of a $k\times n$ matrix $M$ can be computed based on the 
    number of occurrences of each tuple $c\in \sem S^k$ as a column in $M$, as follows.
Fix a matrix $k\times k$ matrix $N$.
For each $i\in [k]$, let $d_i$ be the $i$th diagonal entry of $N$, let $c_i$ be the $i$th column of $N$, let 
$n_i$ be the number 
of occurrences of $c_i$ in $N$ and let $m_i$ be the number of occurrences of $c_i$ in $M$.
Finally, let $s_i$ be the product $m_i\cdots (m_i-n_i+1)\cdot d_i\in \sem S$, and let $p_N$ be the product $s_1\cdots s_k$. 
The permanent of $M$ is equal to the sum of all the values $p_N$, for $N$ ranging over all $k\times k$ matrices. Note that the value $s_i$ can be computed by applying $n_i$ times Lemma~\ref{lem:mod}. 
This  yields the required circuit family.
\end{proof}

\subsection{Proof of Theorem~\ref{thm:provenance}}

\begin{proof}[\ifdoublecolumn{Proof }of Lemma~\ref{lem:perm-iter}]
    Let $M$ be an $R\times C$ matrix and let $r\in R$ be an arbitrarily chosen row.
Observe that the following identity holds:
    \begin{align}\label{perm:iter}
        \perm (M)=\sum_{{c\in C}}M[r,c]\cdot \perm (M^{rc}),
    \end{align}
    where $M^{rc}$ is the matrix $M$ with row $r$ and column $c$ removed. The idea is to recursively use the above identity to enumerate $\perm(M)$.
    For this to work, we need to iterate over all columns $c$ such that the product $M[r,c]\cdot \perm (M^{rc})$ is nonempty. 

    To this end,
    observe that the function $h\from \sem F_A\to \sem B$ which maps $0$ to $0$ and every non-zero element in $\sem F_A$ to $1$, is a semiring homomorphism.
    It follows that 
    $h(\perm(A))=\perm(h(A))$,
    for any matrix $A$ with entries in $\sem F_A$ and matrix $h(A)$  obtained from $A$ by mapping each entry via $h$.

    Let $N$ denote the matrix $h(M)$, and let $N^{rc}$ denote the matrix obtained by removing the row $r$ and column $c$.
    Thanks to the above, the problem mentioned earlier boils down to the problem considered in the following lemma.
    
    \begin{lemma}
        \label{lem:trivial-bool}
        There is a dynamic algorithm which inputs an $R\times C$ boolean matrix $N$ and a row $r$,
and computes in linear time a data structure which maintains in constant time a 
bi-directional iterator for the set of those columns $c\in C$ with $N[r,c]= \perm(N^{rc})=1$.
    \end{lemma}

     Finally, to enumerate $\perm(M)$, we use~\eqref{perm:iter}. Fix $r\in R$ and iterate over all columns $c$ such that $N[r,c]= \perm(N^{rc})=1$.
     For  each considered column $c$,  the iterator for $M[r,c]$ is given by assumption, and the iterator for $\perm (M^{rc})$ is obtained recursively (note that $M^{rc}$ has one row fewer than $M$).
     Using them, iterate over the product $M[r,c]\cdot \perm(M^{rc})$ in a lexicographic fashion. 
\end{proof}

\begin{proof}[\ifdoublecolumn{Proof} of Lemma~\ref{lem:trivial-bool}]
    Our data structures stores, for each vector $\tup t\in \sem B^S$,
    a bi-directional list $L_\tup t$ over all occurrences of $\tup t$ as a column in $N$.

    Note that this data structure can computed in time $\Oof_{R}(|C|)$.
    Moreover, it can be 
     maintained in time $\Oof_R(1)$, upon an update to $N$:
     if an update affects an entry at row $r$ and column $c$,
     then we remove this column from  the lists it belonged to,
     and append it to the appropriate list.

     Let $K\subset \sem B^S$ be the 
     set of those vectors $\tup t$ which occur in $N$ as some column $c$ such that $N[r,c]=\perm(N^{rc})=1$. Note that this property does not depend on the choice of the column $c$ among all occurrences of $\tup t$.
     
     Moreover, $K$ can be computed only basing on the number of occurrences of each vector $\tup t\in\sem B^S$ as a column in $N$, where  the number of occurrences is only counted up to $|R|$. 
     It follows that the set $L$ can be computed in constant time, basing on the existing data structure, 
     since we can iterate through $L_{\tup t}$ for at most $|R|$ steps to determine the number of occurrences. 
    
     This yields the desired bi-directional iterator for the columns $c$ such that $N[r,c]= \perm(N^{rc})=1$:
     it is the concatenation of the lists $L_\tup t$, for $\tup t\in K$. 
\end{proof}

\begin{proof}[\ifdoublecolumn{Proof }of Theorem~\ref{thm:provenance}]
    We give the proof in the case when $f$ is a closed expression. The general case follows as in the proof of Theorem~\ref{thm:algo-main},
    by simulating a query by a sequence of $|\tup x|$ temporary updates.
    
    Consider the circuit family $(C_\str A)_{\str A\in \CCC}$ given by Theorem~\ref{thm:circuits}.
For simplicity, and without loss of generality, assume that all inner gates in $C_\str A$ are permanent gates, as multiplication and addition can be simulated by such.

To prove Theorem~\ref{thm:provenance},
 apply  Lemma~\ref{lem:perm-iter}  to each gate $g$ of $C_\str A$, starting from the input gates and moving towards the output gate.
This yields an iterator representing the value at the gate $g$ of $C_\str A$ induced by the weight functions $\tup w$.
The iterator at the output gate is the required iterator for $f_{\str A(\tup w)}$.
\end{proof}

\subsection{Proof of Theorem~\ref{thm:enum}}
We now consider the dynamic case of Theorem~\ref{thm:enum}.

Without loss of generality, assume that the signature $\Sigma$ contains only
relation symbols of arity bounded by some number $r_{\max}\in\N$.
Also without loss of generality, we may assume that the Gaifman graph of the input structure $\str A$
is left unchanged by all the updates. Indeed, we can assume that the input structure $\str A$  includes a binary edge relation $E$ which initially represents exactly the Gaifman graph $G_\str A$ of $\str A$, and which is never modified by updates
(the relation $E$ can be added to $\str A$ in linear time in the preprocessing phase). Furthermore, for all $2\le r\le r_{\max}$, we assume that $\str A$ contains a relation $R_k$ consisting of all tuples $\tup a\in \str A^k$ whose elements form a clique in $G_\str A$, and that the relations $R_1,\ldots,R_r$ are never modified by updates. 
Henceforth, we assume that 
the updates never modify the relations $R_1,\ldots,R_{r_{\max}}$ of a structure $\str A$. 
We call those updates \emph{Gaifman-preserving updates} since  they preserve the Gaifman graph of $\str A$. 

To employ our circuit framework, we encode the relations of $\str A$
by weight functions with values $\set{0,1}\in \sem F$, where $\sem F$ is the free commutative semiring.

\begin{lemma}\label{lem:gaifman-reduction}
    Let $\Sigma,\phi,\CCC$ be as in Theorem~\ref{thm:enum}, with $\phi(\tup x)$ quantifier-free.
    There is a 
    signature  $\Sigma'$,
    a set of weight symbols $\tup v$, a quantifier-free $\Sigma'(\tup v)$-expression $g(\tup x)$,
    and a dynamic algorithm which, given a structure $\str A\in\CCC$, computes in linear time a $\Sigma'$-structure $\str A'$ with the same Gaifman graph as $\str A$ and a tuple of weight functions $\tup v_\str A$ with values in $\sem {0,1}$, which interpreted in any semiring $\sem S$ satisfy 
    \begin{align}
        \label{eq:gaifman-red}
        g_{\str A'(\tup v)}(\tup a)=[\phi_{\str A}(\tup a)]\qquad\text{for \ }\tup a\in \str A^\tup x.
    \end{align}
    When a Gaifman-preserving update is applied to $\str A$, then $\str A'(\tup v)$ is updated in constant time by updating $\tup v$ and leaving $\str A'$ unchanged.
\end{lemma}

\begin{proof}
    
Given a structure $\str A$, the structure $\str A'$ 
contains only the relations $R_1,\ldots,R_{r_{\max}}$ as described before the lemma. In particular, $\str A'$ is not modified when a Gaifman-preserving update is applied to $\str A$. 

    Define a set of weight symbols $\tup v$ containing 
    symbols $v_R^+$ and $v_R^-$ of arity $r$, for each 
    relation $R\in \Sigma$ of arity $r$. For each 
    $R\in\Sigma$ of arity $r$, define $\sem S$-weight functions 
    $v_R^+$ and $v_R^-$ of arity $r$ on $\str A'$ as 
    follows:
    \begin{align*}
    v_R^+(a_1,\ldots,a_k)&=[R_\str A(a_1,\ldots,a_k)], \\v_R^-(a_1,\ldots,a_k)&=[\neg R_\str A(a_1,\ldots,a_k)],
    \end{align*}
    It is clear that given $\str A$, the structure $\str A'(\tup v)$ can be computed in linear time, and updated in constant time, upon  Gaifman-preserving updates to $\str A$.
    
    It remains to describe the expression $g(\tup x)$ satisfying~\eqref{eq:gaifman-red}.
    Let $T$ be the set of terms occurring in $\phi$. 
    Rewrite $\phi(\tup x)$ as a disjunction of atomic types $\alpha(\tup x)$, 
    where an \emph{atomic type} is a maximal consistent conjunction of literals involving terms from $T$ only. 
    As any two inequivalent atomic types $\alpha(\tup x)$ and $\beta(\tup x)$ are mutually exclusive,
    it follows that $[\alpha(\tup a)\lor \beta(\tup a)]=[\alpha(\tup a)]+[\alpha(\tup a)]$, in any semiring $\sem S$.
    Therefore, it is enough to consider the case when $\phi(\tup x)$ is an atomic  type $\phi(\tup x)=\alpha_1(t_1)\land \cdots \land\alpha_k(t_k)$,
    where $\alpha(t_i)$ are literals. 
    Clearly, 
    no matter what semiring $\sem S$ is considered,
    the following equivalence holds:
    $$[\alpha_1(t_1)\land\cdots\land \alpha_k(t_k)]=[\alpha_1(t_1)]\cdots [\alpha_k(t_k)].$$
    Therefore, it is enough to consider the case when $\phi(\tup x)$ is a single literal $R(t^1,\ldots,t^k)$ or $\neg R(t^1,\ldots,t^k)$, for some terms $t^1,\ldots,t^k\in T$.
    In the first case, define $g(\tup x)$ as $v_R^+(t^1,\ldots,t^k)$,
    and in the latter,  define $g(\tup x)$ as $v_R^-(t^1,\ldots,t^k)$.
    By construction of $\str A'(\tup v)$, it is clear that~\eqref{eq:gaifman-red} holds. 



    
    \end{proof}

\begin{proof}[\ifdoublecolumn{Proof} of Theorem~\ref{thm:enum}]
    We proceed similarly to the static case described in Section~\ref{sec:enum}.
We use Theorem 6.3 in~\cite{Dvorak2013testing},
which extends Theorem~\ref{thm:qe}
by stating that $\whstr A$ can be updated in constant time upon a Gaifman-preserving update to $\str A$.
Thanks to this, we can assume that $\phi(\tup x)$ is quantifier-free.

Let $\tup x=\set{x_1,\ldots,x_k}$.
As in the static case, define weight functions $w_1,\ldots,w_k\from \str A\to \sem F$, where $w_i(a)=e_a^i$.
    Instead of considering the expression $f$  defined in~\eqref{eq:rrr}, we now define  $f$ as follows:
    \begin{align}\label{eq:sss}
        f=\sum_{\tup x}g(\tup x)\cdot w_1(x_1)\cdots w_k(x_k),
    \end{align}
    where $g$ is given by Lemma~\ref{lem:gaifman-reduction}.
    Our dynamic data structure maintains the $\Sigma'(\tup v)$-structure $\str A'(\tup v)$ as described in the lemma.
In particular, the invariant $g_{\str A'(\tup v)}(\tup x)=[\phi_\str A(\tup x)]$ is maintained by updates.

Together with the weight functions $w_1,\ldots,w_k$, this yields a $\Sigma'(\tup v\tup w)$-structure $\str A'(\tup v\tup w)$.
By~\eqref{eq:sss}, 
the element $f_{\str A'(\tup v\tup w)}\in \sem F$ is a formal sum representing the set of tuples $\phi_\str A\subset \str A^{\tup x}$, and this is maintained by the updates. 
Applying Theorem~\ref{thm:provenance} yields a bi-directional iterator with constant access time which enumerates the components of $f_{\str A'(\tup v\tup w)}$.
More precisely, our dynamic algorithm is obtained by composing the dynamic algorithm maintaining $\str A'(\tup v\tup w)$ given by Lemma~\ref{lem:gaifman-reduction} with the dynamic data structure maintaining an enumerator for $f_{\str A'(\tup v\tup w)}$ given by Theorem~\ref{thm:provenance}. An update to $\str A$ triggers constantly many updates to $\tup v$, and an enumerator for $f_{\str A'(\tup v\tup w)}$  can be maintained in constant time.
\end{proof}

\subsection{Proof of Theorem~\ref{thm:foc}}
\begin{proof}

We proceed by induction on the number of guarded connectives in $\phi$
of the form $[R(x_1,\ldots,x_l)]\cdot c(\phi_1,\ldots,\phi_k)$.
In the  base case, there are no guarded connectives.
If $\phi$ is has values in $\sem B$, it is a first-order formula, and the statement follows from Theorem~\ref{thm:qe} and Theorem~\ref{thm:enum} for the enumeration part.
Otherwise, $\phi$ is an $\sem S$-valued formula,
for some semiring $\sem S\in\ops$, and
$\phi$ can be seen as  a $\Sigma'(\tup w)$-expression, where $\Sigma'$ consists of the function symbols and the $\sem B$-valued relations in $\Sigma$, and $\tup w$ consists of the $\sem S$-valued relations in $\Sigma$. In this case, the statement follows from Theorem~\ref{thm:algo-main}.

In the inductive step, let $\phi(\tup x)\in\foops$ and suppose that the statement holds for every formula with fewer guarded connectives.

Let $[R(x_1,\ldots,x_l)]\cdot c(\phi_1,\ldots,\phi_k)$ be a guarded connective in $\phi$ which does not appear in the scope of another guarded connective, and suppose $c\from \sem S_1\times\cdots \sem S^k\to \sem S$. Let $\tup x=\set{x_1,\ldots,x_l}$.
Apply the inductive assumption to the formulas $\phi^1(\tup x),\ldots,\phi^k(\tup x)$. Given a structure $\str A\in \CCC$,
we can compute the functions $\phi^i_\str A\from\str A^\tup x\to \sem S_i$ in the required time.
Next, define the
$\sem S$-valued relation $r\from \str A^k\to\sem S$  by $r(\tup a)=c(\phi^1_\str A(\tup a),\ldots, \phi^k_\str A(\tup a))$ for $\tup a\in  R_\str A\subset \str A^\tup x$, and $r(\tup a)=0$ for $\tup a\not\in R_\str A$.
This can be computed in linear time, by scanning through all tuples $\tup a\in R_\str A$ and for each of them,  applying in constant time the connective $c$ to the tuple of precomputed values. 

Define the structure $\str A'$ as $\str A$ extended by  the $\sem S$-valued relation $r$.
Replace in the formula $\phi(\tup x)$  the guarded connective $[R(x_1,\ldots,x_l)]\cdot c(\phi_1,\ldots,\phi_k)$
by the atom $r(x_1,\ldots,x_l)$. Clearly, $\phi'_{\str A'}=\phi_\str A$. Moreover, $\phi'$ has fewer guarded connectives than $\phi$, so the conclusion follows from the inductive assumption.
\end{proof}

\end{document}